\documentclass[11pt,a4paper,reqno]{amsart}

\usepackage[english]{babel}
\usepackage[latin1]{inputenc}
\usepackage[hidelinks]{hyperref}
\usepackage{amsmath,
math rsfs,
xcolor,
dsfont,
mathabx
}

\usepackage{geometry}
\makeatletter
\renewcommand*{\@textcolor}[3]{%
  \protect\leavevmode
  \begingroup
    \color#1{#2}#3%
  \endgroup
}
\makeatother

\newtheorem{teo}{Theorem}[section]
\newtheorem{lem}{Lemma}[section]
\newtheorem{pro}{Proposition}[section]

\theoremstyle{remark}
\newtheorem{rem}{Remark}[section]

\newcommand{\R}{\mathbb{R}}
\newcommand{\C}{\mathbb{C}}

\newcommand{\flap}{(-\Delta)^s}
\newcommand{\flaph}{(-\Delta)^{s/2}}

\newcommand{\dom}{\mathcal{D}}
\newcommand{\ham}{\mathcal{H}_s^\alpha}
\newcommand{\hamz}{\mathcal{H}_s^0}
\newcommand{\hamn}{\mathcal{H}_s^{n}}
\newcommand{\form}{\mathcal{F}_s^\alpha}
\newcommand{\formz}{\mathcal{F}_s^0}
\newcommand{\propz}{\mathcal{U}_s}
\newcommand{\app}{\mathcal{H}_s^\varepsilon}
\newcommand{\green}{\mathcal{G}_s^\lambda}
\newcommand{\greeno}{\mathcal{G}_s^\omega}
\newcommand{\sgn}{\mathrm{sgn}}
\newcommand{\ep}{\varepsilon}
\newcommand{\LL}{\mathcal{L}}

\newcommand{\tr}[1]{\widehat{{#1}}}
\newcommand{\f}[2]{\frac{#1}{#2}}
\newcommand{\tf}[2]{\tfrac{#1}{#2}}
\newcommand{\re}[1]{\mathrm{Re}\left\{#1\right\}}
\newcommand{\im}[1]{\mathrm{Im}\left\{#1\right\}}
\newcommand{\ov}[1]{\overline{#1}}

\newcommand{\dx}{\,dx}
\newcommand{\dy}{\,dy}
\newcommand{\dk}{\,dk}
\newcommand{\dt}{\,dt}
\newcommand{\dz}{\,dz}
\newcommand{\dpi}{\,dp}
\newcommand{\dtau}{\,d\tau}
\newcommand{\dro}{\,d\rho}
\newcommand{\dome}{\,d\omega}
\newcommand{\deta}{\,d\eta}

\usepackage{graphicx,xcolor}

\usepackage[deletedmarkup=xout,final
]{changes}

\definechangesauthor[name={Raffaele}, color={red}]{Raffaele}
\newcommand{\rafa}[1]{\added[id=]{#1}}
\newcommand{\rafd}[1]{\deleted[id=]{#1}}
\newcommand{\rafr}[2]{\replaced[id=]{#1}{#2}}

\definechangesauthor[name={Lorenzo}, color={blue}]{Lorenzo}
\newcommand{\lora}[1]{\added[id=]{#1}}
\newcommand{\lord}[1]{\deleted[id=]{#1}}
\newcommand{\lorr}[2]{\replaced[id=]{#1}{#2}}

\definechangesauthor[name={Domenico}, color={green}]{Domenico}
\newcommand{\doma}[1]{\added[id=]{#1}}

\begin{document}
 
 \title[Nonlinear singular perturbations of the FSE in d=1]{Nonlinear singular perturbations of the \\ fractional Schr\"odinger equation in dimension one}

\author[R. Carlone]{Raffaele Carlone}
\address{Universit\`{a} degli Studi di Napoli ``Federico II'', Dipartimento di Matematica e Applicazioni ``R. Caccioppoli'', MSA, via Cinthia, I-80126, Napoli, Italy.}
\email{raffaele.carlone@unina.it}
\author[D. Finco]{Domenico Finco}
\address{Universit\`a Telematica Internazionale Uninettuno, Facolt\`a di Ingegneria, corso Vittorio Emanuele II, 00186, Roma, Italy.}
\email{d.finco@uninettunouniversity.net}
\author[L. Tentarelli]{Lorenzo Tentarelli}
\address{Universit\`{a} degli Studi di Napoli ``Federico II'', Dipartimento di Matematica e Applicazioni ``R. Caccioppoli'', MSA, via Cinthia, I-80126, Napoli, Italy.}
\email{lorenzo.tentarelli@unina.it}

\date{\today}

\begin{abstract} 
 The paper discusses nonlinear singular perturbations of delta type of the fractional Schr\"odinger equation $\imath\partial_t\psi=\flap\psi$, with $s\in(\f{1}{2},1]$, in dimension one. Precisely, we investigate local and global well posedness (in a strong sense), conservations laws and existence of blow-up solutions and standing waves.
\end{abstract}

\maketitle


\section{Introduction}
The appearance of the term ``fractional Schr\"odinger equation'' (FSE) traces back to  the pioneering papers  by Laskin (\cite{L1,L2}). In those papers this equation, i.e.
\begin{equation}
 \label{eq-free}
 \imath\partial_t\psi=\flap\psi,
\end{equation}
turns out when  the Feynman path integral is extended from the Brownian-like to the L\'evy-like quantum mechanical paths. On the other hand, in potential theory, the fractional \rafr{Laplacian}{laplacian} was also introduced as the infinitesimal generator of some L\'evy processes (\cite{BE,LA}).

In the most recent physics literature the fractional nonlinear Schr\"odinger  equation (FNSE) have been exploited in many different frameworks.

The FNSE arises, for instance, in the investigation of quantum effects in Bose-Einstein Condensation (\cite{UB}). The condensation is usually described by a Gross-Pitaevskii equation when a boson system is ideal, i.e. at low temperature and weakly interacting. If the medium is inhomogeneous  with long-range  interactions, it was proved that the dynamics can \rafa{be} described by a fractional Gross-Pitaevskii equation. This phenomenon is due to a general feature of fractional equations, which are able to take into the account decoherence and turbulence effects. In fact, \rafr{these}{this} equations can be also used to investigate decoherence effects in several physical models (\cite{KZ}). 

However, physical models based on the FNSE as effective equation are not restricted to the quantum domain. The study of the long-time behavior of the solutions of the water waves equation in $\mathbb{R}^{2}$ relies, for instance, on the FNSE (see \cite{IP} and references therein). Furthermore, it has been recently used in optics (\cite{L}), in solid state physics (\cite{PV}) and biological systems (\cite{DV,MV}).

\medskip
From the mathematical point of view, the FNSE has been widely studied in the last years. We mention, for instance, \cite{BHL,GH,KLR} for the standard FNSE and \cite{CHH1,CHKL1,CHKL2,HY} for the variant (relevant for physical applications) provided by Hartree-type nonlinearities.

On the other hand, the first discussion of a FSE in presence of delta potentials is due to \cite{LRSRM}: the main result contained in this paper is the analysis of an anomalous spreading of the solutions of the Cauchy problem characterized by a power-law behavior and related to the order of the fractional operator.

More recently \cite{MOS,MS,S} addressed the problem of  singular perturbations of the fractional \rafr{Laplacian}{laplacian}. In particular, they show that rank-one linear singular perturbations of the $d$-dimensional fractional Laplacian can be obtained as the norm-resolvent limit of fractional Schr\"odinger operators with shrinking potentials.

\medskip
In this paper we discuss \emph{nonlinear} singular perturbations of delta type of \eqref{eq-free} in dimension one, with $s\in(\f{1}{2},1]$. The main purpose is to investigate the well-posedness and the dynamical features of the associated Cauchy problem.

The major characteristic  of this kind of equations is that, albeit nonlinear, they fall under the so-called \emph{solvable models}: the investigation of the time evolution can be reduced to that of an ODE-type equation (see, e.g., \cite{AGH-KH}). 
For the ordinary \rafr{Laplacian}{laplacian} in $\mathbb{R}$ the case of a concentrated nonlinearity has been studied in \cite{AT} (see also \cite{CFN}) and then extended to $\mathbb{R}^{2}$ and $\mathbb{R}^{3}$ in \cite{ACCT,CCF,CCT,CFT2} and \cite{ADFT,ADFT2} (respectively), and to the $1$-dimensional Dirac equation in \cite{CCNP}.

\medskip
The paper is organized as follows:
\begin{itemize}
 \item in Section \ref{sec:setting} we make a brief review on the linear FSE in the free case and in the case of a singular perturbation of delta type, and then we present the main results on nonlinear perturbations;
 \item in Section \ref{sec:proof} we present the proofs of the main results of the paper (i.e., Theorem \ref{teo-nonlinear} and Theorem \ref{teo-standing});
 \item in Appendix \ref{sec:linear} we show the proof of some further results on linear perturbations (i.e., Proposition \ref{pro-linear}).
\end{itemize}


\section{Setting and main results}
\label{sec:setting}

Before stating the main results of the paper it is necessary to fix some notation and recall some well known facts about the linear FSE both in the free case and in the case of delta interactions.


\subsection{The free case}

Preliminarily, we recall that, for every $s\in(0,1)$, the \emph{fractional} Laplace operator on (an open set) $\Omega\subset\R$ is defined by
\[
 (-\Delta)_{\Omega}^s u(x):=c(s)P.V.\int_\Omega\f{u(x)-u(y)}{|x-y|^{1+2s}}\dy
\]
where $c(\,\cdot\,)$ is a specific normalization constant such that
\begin{equation}
 \label{eq-flaptr}
 \tr{\flap u}\,(k):=|k|^{2s}\,\tr{u}(k),
\end{equation}
with $\flap=(-\Delta)_{\R}^s$ and $\tr{(\,\cdot\,)}$ representing the unitary Fourier transform on $\R$ (see, for instance, \cite{DPV} for an explicit expression). We also recall the definition of the fractional Sobolev space $H^\mu(\Omega)$, with $\mu\in(0,1)$, i.e.
\begin{equation}
 \label{eq-defsob1}
 H^\mu(\Omega):=\left\{u\in L^2(\Omega):u\in\dot{H}^\mu(\Omega)\right\}=\left\{u\in L^2(\Omega):[u]_{\dot{H}^{\mu}(\Omega)}<\infty\right\}
\end{equation}
where $[u]_{\dot{H}^{\mu}(\R)}$ denotes the usual Gagliardo semi-norm, i.e.
\[
 [u]_{\dot{H}^{\mu}(\Omega)}^2:=\int_{\Omega\times\Omega}\f{|u(x)-u(y)|^2}{|x-y|^{1+2\mu}}\dx\dy.
\]
When $\Omega=\R$, it is (norm-)equivalent to
\begin{equation}
 \label{eq-defsob2}
 H^\mu(\R)=\left\{u\in \mathcal{S}'(\R):\int_{\R}(1+|k|^2)^{\mu}\,|\tr{u}(k)|^2\dk<\infty\right\},
\end{equation}
where $\mathcal{S}'(\R)$ denotes the space of the tempered distributions and
\begin{equation}
 \label{eq-gagliardo_fou}
 [u]_{\dot{H}^{\mu}(\R)}^2:=\|(-\Delta)^{\mu/2} u\|_{L^2(\R)}^2=\int_{\R}|k|^{2\mu}\,|\tr{u}(k)|^2\dk
\end{equation}
(see again \cite{DPV}). Moreover, \eqref{eq-defsob2} holds for $\mu>1$, as well, whereas \eqref{eq-defsob1} has to be modified in
\[
 H^\mu(\Omega):=\left\{u\in H^{[\mu]}(\Omega):\f{d^{[\mu]}u}{dx^{[\mu]}}\in\dot{H}^{\mu-[\mu]}(\Omega)\right\}
\]
(with $[\mu]$ the integer part of $\mu$).

\medskip
On the other hand, it is well known that, for every $s\in(0,1]$, the operator $\hamz:L^2(\R)\to L^2(\R)$ defined by
\[
 \dom(\hamz):=H^{2s}(\R)\qquad\text{and}\qquad\hamz\psi:=\flap\psi,\qquad\forall\psi\in\dom(\hamz),
\]
is self-adjoint and hence by Stone's Theorem the Cauchy problem
\[
 \left\{
 \begin{array}{l}
  \displaystyle \imath\f{\partial\psi}{\partial t}=\hamz\psi\\[.4cm]
  \displaystyle \psi(0,\,\cdot\,)=\psi_0(\,\cdot\,)\in H^{2s}(\R)
 \end{array}
 \right.
\]
has a unique solution
\[
 \psi(t,x):=(\propz(t)\psi_0)(x)\in C^0([0,T];H^{2s}(\R))\cap C^1([0,T];L^{2}(\R)),\quad \forall T>0,
\]
 with $\propz(t)$ the convolution unitary operator defined by the integral kernel
\[
 \propz(t,x):=\f{1}{2\pi}\int_\R e^{\imath kx}\,e^{-\imath|k|^{2s}t}\dk,
\]
namely
\[
 \tr{\propz(t,\,\cdot\,)}(k)=\f{e^{-\imath|k|^{2s}t}}{\sqrt{2\pi}}.
\]
In addition, we recall that the quadratic form associated with $\hamz$ is
\[
 \dom(\formz):=H^s(\R)\qquad\text{and}\qquad\formz(\psi):=[\psi]_{\dot{H}^{s}(\R)}^2,\quad\forall\psi\in\dom(\formz).
\]

\medskip
Finally, we point out that throughout the paper we denote by $(\,\cdot\,,\,\cdot\,)$ and $\|\,\cdot\,\|$ the standard scalar product and norm of $L^2(\R)$.


\subsection{Linear delta perturbations}
\label{subsec:lin}

The problem of linear singular perturbations of delta type of the fractional \rafr{Laplacian}{laplacian} has been widely discussed in \cite{MS}. Here we limit ourselves to mention some known facts and notation, which are useful in the following.

For every $\lambda>0$ and $s\in(0,1]$, we denote by $\green$ the Green's function of $\flap+\lambda$, namely the unique solution in $L^2(\R)$ of
\begin{equation}
 \label{eq-green_eq}
 (\flap+\lambda)\,\green=\delta,
\end{equation}
namely
\begin{equation}
 \label{eq-green}
 \tr{\green}(k)=\f{1}{\sqrt{2\pi}\,(|k|^{2s}+\lambda)}.
\end{equation}
Note that $\green\in L^2(\R)$ for every $s>1/4$ and $\green\in L^\infty(\R)\cap C^0(\R)$ for every $s>\f{1}{2}$.

\medskip
In addition, for all $s\in(\f{1}{2},1]$ and $\alpha\in\R$ fixed, we define the operator $\ham:L^2(\R)\to L^2(\R)$
\begin{gather}
 \label{eq-dom1}
 \dom(\ham):=\left\{\psi\in L^2(\R):\psi=\phi_\lambda-\alpha\green\psi(0),\,\phi_\lambda\in H^{2s}(\R),\,\lambda>0\right\},\\[.3cm]
 \label{eq-act1}
 (\ham+\lambda)\psi:=(\flap+\lambda)\phi_\lambda,\qquad\forall\psi\in\dom(\ham).
\end{gather}
It is self-adjoint and hence
\begin{equation}
 \label{eq-cauchy_lin}
 \left\{
 \begin{array}{l}
  \displaystyle \imath\f{\partial\psi}{\partial t}=\ham\psi\\[.4cm]
  \displaystyle \psi(0,\,\cdot\,)=\psi_0(\,\cdot\,)\in\dom(\ham)
 \end{array}
 \right.
\end{equation}
has a unique solution
\[
 \psi\in C^0([0,T];\dom(\ham))\cap C^1([0,T];L^{2}(\R)),
\]
where $\dom(\ham)$ is endowed with the graph norm. Notice that the definition of the operator domain \eqref{eq-dom1} does not depend on $\lambda$, \lora{provided that $\lambda>0$}, which then plays the role of a mere regularizing parameter.

\medskip
It is also worth mentioning that \cite{MS} provides an approximation result for the operator $\ham$, which actually explains the reason for which we consider this operator a singular perturbation of the fractional \rafr{Laplacian}{laplacian} of delta type. Precisely, the paper shows that there exist (a wide class of) compactly supported potentials $V$, with $\alpha=\int_\R V(x)\dx$, such that the sequence of self-adjoint operators
\[
 \dom(\app):=H^{2s}(\R)\qquad\text{and}\qquad\app\,\psi:=\left(\flap+\tf{1}{\ep}V\left(\tf{\cdot}{\ep}\right)\right)\psi
\]
converges in the norm resolvent sense to $\ham$. Therefore, $\ham$ defined by \eqref{eq-dom1} and \eqref{eq-act1} provides a rigorous version of the informal expression $\flap+\delta$.

\begin{rem}
 It is clear that the case $\alpha=0$ represents the unperturbed case.
\end{rem}

Finally, we provide here a different representation of the operator $\ham$. Let us introduce the fractional differential operator
\begin{equation}
 \label{eq-der}
 D^\mu u\,(x):=\f{\imath}{\sqrt{2\pi}}\int_\R e^{ikx}\,|k|^\mu\,\sgn(k)\,\tr{u}(k)\dk,\qquad\mu\in(0,1],
\end{equation}
that is
\[
 \tr{D^\mu u}\,(k)=\imath|k|^\mu\,\sgn(k)\,\tr{u}(k).
\]
Note that this operator is bounded $H^{\mu+\delta}(\R)\to H^{\delta}(\R)$, for all $\delta\geq0$.

\begin{rem}
 Note that, when $\mu=1$, there results $D^{\mu}\equiv\frac{d}{dx}$ (in a distributional sense).
\end{rem}

\begin{pro}
 \label{pro-linear}
 Let $s\in(\f{1}{2},1]$ and $\alpha\in\R$. Then, the following representation of the operator $\ham$ is equivalent to \eqref{eq-dom1}-\eqref{eq-act1}:
  \begin{gather}
  \label{eq-dom2}
  \dom(\ham)=\left\{\psi\in H^s(\R):D^{2s-1}\psi\in H^1(\R\backslash\{0\}),\,[D^{2s-1}\psi](0)=\alpha\psi(0)\right\}\\[.3cm]
  \label{eq-act2}
  \ham\psi=\flap\psi,\qquad\forall x\neq0,
 \end{gather}
 where $[D^{2s-1}\psi](0):=D^{2s-1}\psi(0^+)-D^{2s-1}\psi(0^-)$. Moreover, the quadratic form associated to $\ham$ is defined as
 \begin{equation}
  \label{eq-form}
  \form(\psi):=(\psi,\ham\psi)=\|(-\Delta)^{s/2}\psi\|^2+\alpha|\psi(0)|^2=[\psi]_{\dot{H}^{s}(\R)}^2+\alpha|\psi(0)|^2
 \end{equation}
 with domain $\dom(\form):=H^s(\R)$.
\end{pro}

For the proof, see Appendix \ref{sec:linear}. \lora{Here we limit ourselves to observe that the previous Proposition allows to write also in the fractional case delta perturbations as \emph{``jump conditions''}. However, while in the integer case this fact arises as a natural consequence of the definition of the action and the domain of the operator, in this case it is required to detect (through a careful analysis of the properties of the Green's function -- see Appendix \ref{sec:linear}) the proper fractional differential operator whose jump encodes the delta interaction.} 


\subsection{Main results: nonlinear delta perturbations}
\label{subsec:main}

Let us introduce the nonlinear \rafr{analogue}{analogous} of \eqref{eq-cauchy_lin} by posing
\[
 \alpha=\alpha(\psi)=\beta|\psi(0)|^{2\sigma},\qquad\sigma>0,\quad\beta\in\R,
\]
in \eqref{eq-dom1}, that is
\begin{gather}
 \label{eq-domn1}
 \dom(\hamn):=\left\{\psi\in H^s(\R):\psi=\phi_\lambda-\beta|\psi(0)|^{2\sigma}\psi(0)\green,\,\phi_\lambda\in H^{2s}(\R),\,\lambda>0\right\},\\[.3cm]
 (\hamn+\lambda)\psi:=(\flap+\lambda)\phi_\lambda,\qquad\forall\psi\in\dom(\hamn)\nonumber.
\end{gather}
On the other hand, in view or Proposition \ref{pro-linear}, $\hamn$ can be equivalently represented as
\begin{gather*}
 \dom(\hamn):=\left\{\psi\in H^s(\R):D^{2s-1}\psi\in H^1(\R\backslash\{0\}),\,[D^{2s-1}\psi](0)=\beta|\psi(0)|^{2\sigma}\psi(0)\right\},\\[.3cm]
 \hamn\psi=\flap\psi,\qquad\forall x\neq0.
\end{gather*}

A solution of the nonlinear \rafr{analogue}{analogous} of \eqref{eq-cauchy_lin} is a function $\psi$, such that $\psi(t,\,\cdot\,)\in\dom(\hamn)$, for all $t>0$, and 
\begin{equation}
 \label{eq-cauchy}
 \left\{
 \begin{array}{l}
  \displaystyle \imath\f{\partial\psi}{\partial t}=\hamn\psi\\[.4cm]
  \displaystyle \psi(0,\,\cdot\,)=\psi_0(\,\cdot\,)\in\dom(\hamn).
 \end{array}
 \right.
\end{equation}
In fact, we discuss the integral form of the problem, provided by the Duhamel formula, which reads
\begin{equation}
 \label{eq-Duahmel}
 \psi(t,x)=(\propz(t)\psi_0)(x)-\imath\beta\int_0^t\propz(t-\tau,x)|\psi(\tau,0)|^{2\sigma}\psi(\tau,0)\dtau.
\end{equation}
Furthermore, one can set
\begin{equation}
 \label{eq-charge_def}
 q(t):=\psi(t,0),
\end{equation}
which is usually called \emph{charge} and, from \eqref{eq-Duahmel}, has to satisfy
\begin{equation}
 \label{eq-charge_eq}
 q(t)+\imath a(s)\beta\int_0^t\f{|q(\tau)|^{2\sigma}q(\tau)}{(t-\tau)^{\f{1}{2s}}}\dtau=f(t)
\end{equation}
(i.e. a singular nonlinear Volterra integral equation of the second kind), where
\[
 f(t):=(\propz(t)\psi_0)(0)
\]
and
\begin{equation}
 \label{eq-a}
 \lord{\propz(t,0)=\f{a(s)}{t^{\f{1}{2s}}},\qquad }a(s):=\f{1}{2\pi}\int_\R e^{-\imath|\rho|^{2s}}\dro=\propz(1,0)\in\C
\end{equation}
\lora{(note that $\propz(t,0)=a(s)/t^{\f{1}{2s}}$)}. As a consequence, if \eqref{eq-charge_eq} uniquely determines $q$, then $\psi$ is defined by
\begin{equation}
 \label{eq-ansatz}
 \psi(t,x):=(\propz(t)\psi_0)(x)-\imath\beta\int_0^t\propz(t-\tau,x)|q(\tau)|^{2\sigma}q(\tau)\dtau.
\end{equation}
This entails that ``solving the nonlinear \lorr{analogue}{analogous} of \eqref{eq-cauchy_lin}'' means searching for a solution of \eqref{eq-charge_eq} such that \eqref{eq-ansatz} satisfies \eqref{eq-cauchy}.

\begin{teo}[Well posedness]
 \label{teo-nonlinear}
 Let $s\in(\f{1}{2},1]$ and $\psi_0\in\dom(\hamn)$. Then:
 \begin{itemize}
  \item[(i)] \emph{Local well-posedness.} There exists $T>0$ such that the function $\psi$ defined by \eqref{eq-charge_eq}-\eqref{eq-ansatz} is the unique solution of \eqref{eq-cauchy} (where the former is meant as an equality in $L^2(\R)$). In addition,
  \[
   \psi\in C^0([0,T];\dom(\hamn))\cap C^1([0,T];L^{2}(\R)).
  \]
  \item[(ii)] \emph{Conservation laws.} The mass and the energy, namely
  \begin{gather}
   M(t)=M(\psi(t,\,\cdot\,)):=\|\psi(t,\,\cdot\,)\|\nonumber\\[.1cm]
   \label{eq-energy} E(t)=E(\psi(t,\,\cdot\,)):=[\psi(t,\,\cdot\,)]_{\dot{H}^{\rafr{s}{s/2}}(\R)}^2+\frac{\beta}{\sigma+1}|\psi(t,0)|^{2\sigma+2},
  \end{gather}
 are preserved quantities along the flow.
 \item[(iii)] \emph{Global well-posedness.} If one of the following conditions is satisfied:
 \begin{itemize}
  \item[-] $\beta\geq0$,
  \item[-] $\beta<0$ and $\sigma<\sigma_c(s):=2s-1$,
 \end{itemize}
 then the solution is global in time. In addition, when $\beta<0$ and $\sigma=\sigma_c(s)$ there exists $C(s,\beta)>0$ such that, if $\|\psi_0\|<C(s,\beta)$, then the solution is global in time, as well.
 \item[(iv)] \emph{Blow-up solutions.} If $\beta<0$, $\sigma\geq\sigma_c(s)$ and $\psi_0$ has the regular part $\phi_{\lambda,0}$ in $\mathcal{S}(\R)$ and satisfies
 \[
  E(\psi_0)<0,
 \]
 then there exists $T^*\in[0,\infty)$ such that
 \[
  \limsup_{t\to T^*}|q(t)|=+\infty,
 \]
 namely the solution blows-up in a finite time.
 \end{itemize}
\end{teo}

\begin{rem}
 Recall that the case $\beta>0$ is usually called \emph{defocusing} or \emph{repulsive} case, whereas the case $\beta<0$ is usually called \emph{focusing} or \emph{attractive} case, in analogy with the standard NLS. In addition, when $\beta<0$, the case $\sigma<\sigma_c(s)$ is said \emph{sub-critical}, the case $\sigma=\sigma_c(s)$ is said \emph{critical} and the case $\sigma>\sigma_c(s)$ is said \emph{super-critical}. According to this, $\sigma_c(s)$ is called \emph{critical exponent}.
\end{rem}

\begin{rem}
\rafd{ Observe that, for $\sigma=0$, items (i)--(iii) of Theorem \ref{teo-nonlinear} provide an explicit representation for the solution of \eqref{eq-cauchy_lin}}.
\end{rem}

\begin{rem}
 Notice that, if $\|\psi_0\|<C(s,\beta)$, then $E(0)\geq0$ (see \eqref{eq-epos}) and thus (consistently) the blow-up condition is not fulfilled.
\end{rem}

\begin{rem}
 Throughout the paper we consider the case of an interaction based at $x=0$. However, this is not restrictive since all the results are valid as well as for any other choice. In addition, one could also consider the case of $N$ distinct singular perturbations of $\delta$-type, as well as more general nonlinear dependence for $\alpha(\psi)$. Anyway, this would produce only computational issues and hence, for the sake of simplicity, we limit ourselves to the case of a single perturbation with a power-type nonlinearity.
\end{rem}

\rafr{Two comments on the item (iv) of the previous theorem are in order. First, as highlighted in Remark \ref{rem-second_der}, the extra assumption on the smoothness of the regular part of the initial datum is just a technical point required in the proof of Proposition \ref{pro-second_der} . 
In addition, blow-up is a phenomenon linked to the magnitude of the initial singular part rather than to the smoothness of the regular part. Hence, it is not that relevant a detailed discussion of the minimal regularity assumptions.}{Two comments on  item (iv) of the previous theorem are in order. First, the extra assumption on the smoothness or the regular part of the initial datum is just a technical point which is required in the proof of Proposition \ref{pro-second_der} (as highlighted in Remark \ref{rem-second_der}). This could be easily weakened, but by means of some ad hoc assumptions which are not particularly relevant in this context. In addition, blow-up is a phenomenon linked to the magnitude of the initial singular part rather than to the smoothness of the regular part. Hence, it is not that relevant a detailed discussion of the minimal regularity assumptions.}

On the other hand, it is clear that in the super-critical focusing case, it is possible to choose an initial datum with $E(0)<0$. For instance, consider a generic function $u\in\dom(\hamn)$ with regular part in the class of Schwartz functions and such that $u(0)\neq0$. Then, for every $\nu>0$, define
\[
 u_\nu(x)=u(\nu x).
\]
Recalling that $\tr{u(\nu\,\cdot\,)}(k)=\f{1}{\nu}\tr{u}(\f{k}{\nu})$, a simple scaling argument shows that (in the focusing case)
\[
 E(u_\nu)=\nu^{2s-1}[u]_{\dot{H}^{\rafr{s}{s/2}}(\R)}^2-\frac{|\beta|}{\sigma+1}|u(0)|^{2\sigma+2}
\]
and hence, if $\nu$ is sufficiently small, then $E(u_\nu)<0$.

Before stating the second result, we recall that a \emph{standing wave} is a function $\psi^\omega(t,x)=e^{\imath\omega t}u^\omega(x)$, with $\omega\in\R$ and $0\not\equiv u^\omega\in\dom(\hamn)$, that satisfies \eqref{eq-cauchy}, namely such that
\begin{equation}
 \label{eq-stationary}
 \hamn u^\omega+\omega u^\omega=0.
\end{equation}

\begin{teo}[Standing waves]
 \label{teo-standing}
 Let $s\in(\f{1}{2},1]$. Every positive standing wave of \eqref{eq-cauchy} is given by
 \begin{equation}
  \label{eq-standing}
  u^{\omega}(x):=|\beta|^{\f{1}{2\sigma}}\left(2s\sin\left(\tf{\pi}{2s}\right)\right)^{\f{2\sigma+1}{2\sigma}}\omega^{\f{(2s-1)(2\sigma+1)}{4s\sigma}}\greeno(x),
 \end{equation}
 with $\omega>0$ and $\beta<0$. Any other standing wave with frequency $\omega$ differs from $u^\omega$ in \eqref{eq-standing} by a constant phase factor.
 
 In addition:
 \begin{itemize}
  \item[(i)] if $\sigma<\sigma_c(s)$, then $E(u^\omega)<0$, for all $\omega>0$, and $\inf_{\omega\in\R^+}E(u^\omega)=-\infty$
  \item[(ii)] if $\sigma>\sigma_c(s)$, then $E(u^\omega)>0$, for all $\omega>0$, and $\inf_{\omega\in\R^+}E(u^\omega)=0$
  \item[(iii)] if $\sigma=\sigma_c(s)$, then $E(u^\omega)=0$.
 \end{itemize}
\end{teo}

Theorem \ref{teo-standing} has some important consequences. First, it states that no standing wave may exist in the defocusing case. On the other hand, one can immediately sees that in the supercritical focusing case the infimum level of the energies of the standing waves is exactly the threshold under which there is blow-up.

\medskip
\lora{The assumption $s\in(\f{1}{2},1]$ in Theorem \ref{teo-nonlinear} and Theorem \ref{teo-standing} is explained as follows. First, one can observe that linear $\delta$-type perturbations of the fractional Laplacian do exist only if $s>\f{1}{4}$, since otherwise the Green's function of $\flap+\lambda$ is not square integrable (for details see \cite{MOS,MS}). On the other hand, as we stressed in Section \ref{subsec:lin}, when $s\in(\f{1}{2},1)$ the Green's function is also bounded and continuous and this entails that the strategy for the proofs of the main results (up to several technical modifications) retraces the one of \cite{AT}, since the problem still displays the specific features of the so-called models in \emph{codimension one}. On the contrary, when $s\in(\f{1}{4},\f{1}{2}]$, the issue presents the structure of higher codimensional models, such as nonlinear $\delta$-perturbations of the Laplacian in $\R^2$ and $\R^3$ (precisely, $s\in(\f{1}{4},\f{1}{2})$ retraces delta in $3d$, while $s=\f{1}{2}$ retraces delta in $2d$). In these cases the strategy for the proof of global and local well-posedness (and blow-up) is considerably different (for instance, it is no more true that $q(t)=\psi(t,0)$) and, mainly in the $2d$ case (\cite{ACCT,CCT}), requires a more refined analysis of the kernel of the associated charge equation, which is not clear how to adapt to the fractional case at the moment. Hence, we think that in a first work on nonlinear perturbations of the fractional Laplacian, it is worth starting from one-codimensional problems, leaving higher-codimensional ones to forthcoming papers.}

\lora{However, also in this one-codimensional case, although the strategy is analogous, some relevant and interesting differences can be detected between this work and \cite{AT}. First, here we present a finer discussion of the regularizing properties of the integral kernel $t^{-\f{1}{2s}}$. In particular, such discussion allows to prove a far better regularity for the solution of the charge equation, which is subsequently exploited to show that the wave function defined in \eqref{eq-ansatz} is a solution of \eqref{eq-cauchy} in a \emph{strong} sense. In view of this, the papers contains a stronger result (which holds for $s=1$ too) with respect to \cite{AT}, where only the integral formulation and distributional solutions were treated (recall that in \cite{AT} the {main} achievement was a feasible strategy for the investigation of nonlinear delta perturbations).}

\lora{On the other hand, concerning the blow-up analysis, here it is necessary the introduction of a \emph{fractional} moment of inertia, as in the standard fractional NLSE. The study of such a fractional moment of inertia in the context of nonlinear deltas requires several technical modifications in the proofs (which follow, again, the same strategy of the integer problem), which also make the detection of the fractional critical power possible.}

\lora{Finally, in addition to \cite{AT}, the paper contains a complete classification of the standing waves of \eqref{eq-cauchy} and shows the connection between their energy level and the threshold for the blow-up.}

\begin{rem}
 \lora{It is worth recalling that the different strategy used here with respect to the case of the standard NLSE is due to the singular structure of the nonlinearity and to the consequent singular structure of the functions in the ``operator'' domain. Precisely, the Schr\"odinger equation with a nonlinear delta is not a semilinear equation (as the standard NLSE), since the nonlineariy is encoded in the domain of the ``operator'' and not in its action. However, the fact that the problem is the nonlinear version of a so-called \emph{solvable} model allows to trace back to a one-dimensional integral problem (the charge equation), whose discussion (although completely different) then plays the role of Strichartz estimates in the proof of local and global well-posedness for standard NLSE.}
\end{rem}
\begin{rem}
\doma{Notice that even if the fractional Laplacian can exhibit very different behavior with respect to the integer one in the stationary case, see e.g. \cite{DSV17}, this is not the case 
in our evolution problem. For instance $\propz$ shares some dispersive properties with the free propagator $e^{it\Delta}$, namely a $L^1-L^{\infty}$ estimate with a different exponent.
Moreover if we look at equation \eqref{eq-charge_eq}, whose analysis is the heart of this paper, for $s=1$ it is still a Volterra integral equation with a different Abel Kernel,
with similar qualitative behavior.}
\end{rem}


\section{Proof of the main results}
\label{sec:proof}

This section is devoted to the proof of Theorem \ref{teo-nonlinear} and Theorem \ref{teo-standing}. The first four subsections discuss the four items of Theorem \ref{teo-nonlinear}, while the last one is completely focused on Theorem \ref{teo-standing}.


\subsection{Local well-posedness}

In order to prove local well-posedness of \eqref{eq-cauchy} a detailed analysis of the charge equation \eqref{eq-charge_eq} is required. However, we need some preliminary results concerning fractional Sobolev spaces in $d=1$ and the regularizing properties of the $\f{1}{2s}$-Abel integral kernel $t^{-\f{1}{2s}}$.

\begin{lem}
 \label{lem-reg1}
 Let $s\in(\f{1}{2},1]$ and $\ell\in H^\mu(\R)$, with $\mu\geq0$ and $\mathrm{supp}\{\ell\}\subset[0,r]$, for some $r\in\R^+$. Then, the function
 \begin{equation}
  \label{eq-trans_func}
  \LL(t):=\int_0^t\frac{\ell(\tau)}{(t-\tau)^{\f{1}{2s}}}\dtau,\qquad t\in\R,
 \end{equation}
 belongs to $L_{loc}^2(\R)\cap \dot{H}^{\mu+1-\f{1}{2s}}(\R)$.
\end{lem}

\begin{proof}
 Clearly,
 \begin{equation}
  \label{eq-conv1}
  \LL(t)=(h*\ell)(t), \quad \forall t\in\R,\qquad\text{where}\quad h(t):=\f{H(t)}{t^{\f{1}{2s}}}
 \end{equation}
 and $H$ is the Heaviside function. In addition, for every $t\in[0,r]$, $\LL(t)=\LL_r(t)$, where
 \[
  \LL_r(t):=(h_r*\ell)(t), \quad \forall t\in\R,\qquad\text{with}\quad h_r(t):=\f{H(t)-H(t-r)}{t^{\f{1}{2s}}}.
 \]
 As a consequence, simply recalling basic properties of the convolution product
 \[
  \|\LL\|_{L^2(0,r)}\leq\|\LL_r\|_{L^2(\R)}\leq\|h_r\|_{L^1(\R)}\|\ell\|_{L^2(\R)}=r^{1-\f{1}{2s}}\|\ell\|_{L^2(0,r)}<\infty.
 \]
 On the other hand, one can easily see that $\|\LL\|_{L^2(a,b)}<\infty$, for every $a,b\in\R$, since $\LL(t)=0$, for all $t\leq0$, and since, if $b>r$, then $\|\LL\|_{L^2(0,b)}$ can be estimated arguing as before (simply replacing $r$ with $b$).
 
 Therefore, it is left to prove that $\LL\in\dot{H}^{\mu+1-\f{1}{2s}}(\R)$. From \eqref{eq-conv1} and \cite[Eqs. 3.761.4 and 3.761.9]{GR} we see that
 \[
  \tr{\LL}(k)=\sqrt{2\pi}\,\tr{h}(k)\tr{\ell}(k)
 \]
 where
 \[
  \tr{h}(k)=\f{d(s)}{k^{1-\f{1}{2s}}},\qquad d(s):=\f{e^{-\imath\f{\pi}{2}\left(1-\f{1}{2s}\right)}}{\sqrt{2\pi}}\Gamma(1-\tf{1}{2s}).
 \]
 Hence
 \[
  \int_\R\big||k|^{\mu+1-\f{1}{2s}}\tr{\LL}(k)\big|^2\dk\leq C_s\int_\R\big||k|^\mu\tr{\ell}(k)\big|^2\dk<\infty,
 \]
 which concludes the proof.
\end{proof}

\begin{lem}
 \label{lem-reg2}
 Let $s\in(\f{1}{2},1]$ and $\ell\in H^\mu(0,r)$, with $\mu\geq0$ and $r\in\R^+$. \rafr{Then}{Therefore}:
 \begin{itemize}
  \item[(i)] if $\mu\in[0,\f{1}{2})$, then $\LL\in H^{\mu+1-\f{1}{2s}}(0,r)$;
  \item[(ii)] if $\mu\in(\f{1}{2},\f{3}{2})$ and $\ell(0)=0$, then $\LL\in H^{\mu+1-\f{1}{2s}}(0,r)$;
 \end{itemize}
 (with $\LL$ defined by \eqref{eq-trans_func}). In particular,
 \[
  \|\LL\|_{H^{\mu+1-\f{1}{2s}}(0,r)}\leq C_r\|\ell\|_{H^\mu(0,r)}.
 \]
\end{lem}

\begin{rem}
 \label{reg-gain}
 Note that in Lemma \ref{lem-reg2} the regularity gain provided by the $\f{1}{2s}$-Abel kernel $1-\f{1}{2s}\in(0,\f{1}{2}]$, for all $s\in(\f{1}{2},1]$, and is independent of $\mu$. This two features are the core of the bootstrap argument used in the proof of Proposition \ref{prop-charge_reg}.
\end{rem}

\begin{proof}[Proof of Lemma \ref{lem-reg2}]
 Let
 \[
  \widetilde{\ell}(t):=\left\{
  \begin{array}{ll}
   \displaystyle \ell(t),    & \text{if}\quad t\in[0,r],              \\[.1cm]
   \displaystyle \ell(2r-t), & \text{if}\quad t\in(r,2r],             \\[.1cm]
   \displaystyle 0,          & \text{if}\quad t\in\R\backslash[0,2r].
  \end{array}
  \right.
 \]
 Easy computations yield
 \[
  \|\widetilde{\ell}\|_{H^{\mu}(0,2r)}\leq2\|\ell\|_{H^\mu(0,r)},
 \]
 and, on the other hand, \cite[Lemma 2.1]{CFNT2} implies
 \[
  \|\widetilde{\ell}\|_{H^{\mu}(\R)}\leq C\|\widetilde{\ell}\|_{H^{\mu}(0,2r)}.
 \]
 As a consequence $\widetilde{\ell}$ satisfies the assumptions of Lemma \ref{lem-reg1} (with support in $[0,2r]$, instead of $[0,r]$), so that
 \[
  \widetilde{\LL}(t):=\int_0^t\frac{\widetilde{\ell}(\tau)}{(t-\tau)^{\f{1}{2s}}}\dtau\in L_{loc}^2(\R)\cap \dot{H}^{\mu+1-\f{1}{2s}}(\R).
 \]
 Finally, observing that $\widetilde{\LL}(t)=\LL(t)$, for all $t\in[0,r]$, concludes the proof.
\end{proof}

\begin{rem}
 In Lemma \ref{lem-reg2}, $\mu=\f{1}{2}$ is excluded since in this case \cite[Lemma 2.1]{CFNT2} is not valid due to the failure of the Hardy inequality (see also \cite{KP}).  
\end{rem}

\begin{lem}
 \label{lem-ban_alg}
 Let $\ell\in H^\mu(0,r)\cap C^0[0,r]$, with $r\in\R^+$ and $\mu\in[0,1]$. Then, for every $\sigma\geq0$, $|\ell|^{2\sigma}\ell\in H^\mu(0,r)$.
\end{lem}

\begin{proof}
 The cases $\mu=0,1$ are trivial (recalling that $H^0=L^2$). Then, consider the case $\mu\in(0,1)$. Clearly, one sees that
 \[
  \big||\ell(t)|^{2\sigma}\ell(t)-|\ell(\tau)|^{2\sigma}\ell(\tau)\big|\leq C \|\ell\|_{C^0[0,r]}^{2\sigma}|\ell(t)-\ell(\tau)|,\qquad\forall t,\tau\in[0,r],
 \]
 which immediately concludes the proof.
\end{proof}


Now, we can focus on the charge equation. The \rafr{first}{former} step is proving that it admits a unique continuous solution on a sufficiently small interval.

\begin{pro}
 \label{pro-charge_cont}
 Let $s\in(\f{1}{2},1]$ and $\psi_0\in\dom(\hamn)$. Then, there exists $T\in\R^+$ such that \eqref{eq-charge_eq} has a unique solution $q\in C^0[0,T]$.
\end{pro}

\begin{proof}
 Define preliminarily
 \[
  h(t,\tau,q):=\imath a(s)\beta\frac{|q|^{2\sigma}q}{(t-\tau)^{\f{1}{2s}}}.
 \]
 By \cite[Corollary 2.7]{M}, in order to conclude it is sufficient to prove that:
 \begin{itemize}
  \item[(i)] $f$ is continuous on $[0,\infty)$;
  \item[(ii)] for every $\widetilde{t}\in\R^+$ and every bounded set $B\subset\C$, there exists a measurable function $m(t,\tau)$ such that
  \begin{gather*}
   |h(t,\tau,q)|\leq m(t,\tau)\qquad\forall \,0\leq\tau\leq t\leq\widetilde{t},\quad\forall q\in B,\\[.3cm]
   \sup_{t\in[0,\widetilde{t}]}\int_0^tm(t,\tau)\dtau<\infty\qquad\text{and}\qquad\int_0^tm(t,\tau)\dtau\underset{t\to0}{\longrightarrow}0;
  \end{gather*}
  \item[(iii)] for every compact interval $I\subset[0,\infty)$, every continuous function $\varphi:I\to\C$ and every $t_0\in\R^+$
  \[
   \lim_{t\to t_0}\int_I\big(h(t,\tau,\varphi(\tau)-h(t_0,\tau,\varphi(\tau))\big)\dtau=0;
  \]
  \item[(iv)] for every $\widetilde{t}\in\R^+$ and every bounded set $B\subset\C$, there exists a measurable function $n(t,\tau)$ such that
  \begin{gather*}
   |h(t,\tau,q_1)-h(t,\tau,q_2)|\leq n(t,\tau)|q_1-q_2|\qquad\forall \,0\leq\tau\leq t\leq\widetilde{t},\quad\forall q_1,q_2\in B,\\[.3cm]
   n(t,\,\cdot\,)\in L^1(0,t),\quad\forall t\in[0,\widetilde{t}],\qquad\text{and}\qquad\int_t^{t+\ep}n(t+\ep,\tau)\dtau\underset{\ep\to0}{\longrightarrow}0.
  \end{gather*}
 \end{itemize}
 However, (ii)--(iv) can be easily proved setting $m,\,n$ equal to the $\f{1}{2s}$-Abel kernel, up to some suitable multiplicative constants, and exploiting its integrability properties (see, for instance, \cite{ADFT,AT,CCT}).
 
 Hence, it is left to show (i). Preliminarily, note that, as $\psi_0\in\dom(\hamn)$ (and recalling that $\psi_0(0)=q(0)$),
 \begin{equation}
  \label{eq-decomposition}
  f(t)=(\propz(t)\phi_{\lambda,0})(0)-\beta|q(0)|^{2\sigma}q(0)(\propz(t)\green)(0)=:f_1(t)+f_2(t).
 \end{equation}
 Let us discuss the two terms separately. First, simply using Cauchy-Schwarz inequality, we find that for all $T\in\R^+$
 \begin{align}
  \label{eq-estfl2}
  \|f_1\|_{L^2(0,T)}^2 \leq & \, C_s\int_0^T\left(\int_\R|\rafr{\widehat{\phi_{\lambda,0}}(k)}{\phi_\lambda(k)}\dk|\right)^2\dt\nonumber\\[.2cm]
                       \leq & \, C_s\int_0^T\left(\int_\R(1+|k|^2)^s|\rafr{\widehat{\phi_{\lambda,0}}(k)}{\phi_\lambda(k)}|^2\dk\right)\left(\int_\R\f{1}{(|k|^2+1)^s}\dk\right)\dt\nonumber\\[.2cm]
                       \leq & \, C_sT\|\phi_{\lambda,0}\|_{H^s(\R)}^2<\infty.
 \end{align}
 On the other hand, denoting by $\tr{\phi}_p$ the even part of $\tr{\phi}_{\lambda,0}$, and setting $\omega=\rafr{k^{2\, s}}{k^{2\sigma}}$, with some computations one obtains
 \begin{equation}
  \label{eq-f1}
  f_1(t)=\f{1}{\pi}\int_0^\infty e^{-\imath k^{2s}t}\tr{\phi}_p(k)\dk=\f{1}{2\pi s}\int_\R e^{-\imath\omega t}\underbrace{H(\omega)\f{|\omega|^{\f{1}{2s}}\tr{\phi}_p\big(|\omega|^{\f{1}{2s}}\big)}{|\omega|}}_{=:G(\omega)}\dome=\f{\tr{G}(t)}{s\sqrt{2\pi}}=\f{\widecheck{G}(-t)}{s\sqrt{2\pi}}
 \end{equation}
 (where again $H$ denotes the Heaviside function), that is
 \[
  \tr{f}_1(\omega)=\f{G(-\omega)}{s\sqrt{2\pi}}.
 \]
 As a consequence (using the same change of variable as before)
 \begin{align}
  \label{eq-estfh}
  [f_1]_{\dot{H}^{\f{3}{2}-\f{1}{4s}}(\R)}^2 = & \, \f{1}{2\pi s^2}\int_\R|\omega|^{3-\f{1}{2s}}|G(-\omega)|^2\dome\leq C_s\int_0^\infty\f{\omega^{3+\f{1}{2s}}}{\omega^2}\big|\tr{\phi}_p\big(|\omega|^{\f{1}{2s}}\big)\big|^2\dome\nonumber\\[.2cm]
                                          \leq & \, C_s\int_0^\infty k^{4s}|\tr{\phi}_p(k)|^2\dk\leq C_s \int_0^\infty |k|^{4s}|\tr{\phi}_{\lambda,0}(k)|^2\dk<\infty,
 \end{align}
 since $\phi_{\lambda,0}\in H^{2s}(\R)$ by assumption, so that, combining with \eqref{eq-estfl2}, there results $f_1\in H^{\f{3}{2}-\f{1}{4s}}(0,T)$, for all $T\in\R^+$.
 
 Let us consider, now, $f_2$. Precisely, we focus on
 \begin{equation}
  \label{eq-f3}
  f_3(t):=-\frac{f_2(t)}{\beta|q(0)|^{2\sigma}q(0)}.
 \end{equation}
 As $\green\in H^s(\R)$, arguing as \rafa{in} \eqref{eq-estfl2}, one sees that $f_3\in L_{loc}^2([0,\infty))$. On the other hand, arguing as in \eqref{eq-f1}, there results
 \begin{equation}
  \label{eq-f3bis}
  f_3(t)=C_s\int_\R e^{-\imath\omega t}\underbrace{H(\omega)\f{|\omega|^{\f{1}{2s}}}{|\omega|(|\omega|+\lambda)}}_{=:F(\omega)}\dome=C_s\widecheck{F}(-t),
 \end{equation}
 so that
 \begin{equation}
  \label{eq-f3sob}
  [f_3]_{\dot{H}^\mu}^2(\R)=C_s\int_0^\infty\f{\omega^{2(\mu-1)+\f{1}{s}}}{(\omega+\lambda)^2}<\infty,\qquad \forall \mu\in[0,\tf{3}{2}-\tf{1}{2s}).
 \end{equation}
 Hence, as $\f{3}{2}-\f{1}{2s}\in(\f{1}{2},1]$ for $s\in(\f{1}{2},1]$, this implies that $f_3$ is continuous on $[0,\infty)$, which concludes the proof.
\end{proof}

\begin{rem}
 \label{rem-f3l2}
 Equation \eqref{eq-f3bis} also shows that $f_3\in L^2(\R)$, not only in $L_{loc}^2([0,\infty))$.
\end{rem}

The \rafr{last}{latter} step of the discussion of the charge equation is proving a suitable Sobolev regularity for the solution $q$. To this aim, it is also convenient to define the maximal existence interval $[0,T^*)$ for \eqref{eq-charge_eq}, i.e.
\begin{equation}
 \label{eq-maximal}
  T^*:=\sup\{T>0:\text{ there exists a unique solution }q\in C^0[0,T]\text{ of \eqref{eq-charge_eq}}\}.
\end{equation}

\begin{pro}
 \label{prop-charge_reg}
 Let $s\in(\f{1}{2},1]$ and $\psi_0\in\dom(\hamn)$. Then, the solution $q$ of \eqref{eq-charge_eq}, provided by Proposition \ref{pro-charge_cont}, belongs to $H^{\f{3}{2}-\f{1}{4s}}(0,T)$ for every $T\in(0,T^*)$.
\end{pro}

\begin{rem}
 Observe that, as $s\in(\f{1}{2},1]$, then $\f{3}{2}-\f{1}{4s}\in(1,\f{5}{4}]$, so that (in particular) the solution of the charge equation is absolutely continuous on all closed and bounded subintervals of $[0,T^*)$. This justifies any integration of the derivative of the charge present throughout the paper.
\end{rem}

\begin{proof}[Proof of Proposition \ref{prop-charge_reg}]
 Fix an arbitrary $T\in(0,T^*)$. The proof can be divided \rafr{into}{in} two parts.
 
 \emph{Part(i): regularity of the forcing term.} The first point is to show that $f\in H^{\f{3}{2}-\f{1}{4s}}(0,T)$. However, as we already proved in \eqref{eq-estfl2}-\eqref{eq-estfh} that $f_1(t):=(\propz(t)\phi_{\lambda,0})(0)$ belongs to $H^{\f{3}{2}-\f{1}{4s}}(0,T)$, according to the decomposition of $f$ pointed out in \eqref{eq-decomposition}, we focus on $f_2$.
 
 In \eqref{eq-f3sob}, we proved that $f_3$, and whence $f_2$ (see \eqref{eq-f3}), is in $H^\mu(0,T)$ for all $\mu\in[0,\tf{3}{2}-\tf{1}{2s})$, which is not sufficient to our purposes. However, suitably manipulating $f_3$, one can find an equivalent formulation of \eqref{eq-charge_eq}, which presents a forcing term with the proper regularity.
 
 First, letting
 \[
  b(s):=\f{1}{2\pi}\int_\R\f{e^{-\imath|\rho|^{2s}}-1}{|\rho|^{2s}}\dro,
 \]
 we can define
 \[
  \widetilde{f}_3(t):=f_3(t)-f_3(0)-t^{1-\f{1}{2s}}b(s).
 \]
 Using the change of variable $\omega=k^{2s}$ and fundamental theorem of calculus, one finds that
 \[
  \widetilde{f}_3(t)=-\f{\lambda}{\pi}\int_0^\infty\f{e^{-\imath k^{2s}t}-1}{k^{2s}(\rafr{k^{2s}}{K^{2s}}+\lambda)}\dk=-\f{\lambda}{\pi}\int_0^\infty\f{(e^{-\imath \omega t}-1)\omega^{\f{1}{2s}}}{\omega^2(\omega+\lambda)}\dome=\f{\imath\lambda}{\pi}\int_0^\infty\f{\omega^{\f{1}{2s}}}{\omega(\omega+\lambda)}\int_0^t \rafr{e^{-\imath z \omega}}{e^{-\imath zt}} \dz\dome.
 \]
 As a consequence
 \[
  \dot{\widetilde{f}}_3(t)=\f{\imath\lambda}{\pi}\int_0^\infty e^{-\imath\omega t}H(\omega)\f{|\omega|^{\f{1}{2s}}}{|\omega|(|\omega|+\lambda)}\dome
 \]
 and hence, arguing as in \eqref{eq-f3bis}-\eqref{eq-f3sob} (in view of Remark \ref{rem-f3l2}), one has that $\dot{\widetilde{f}}_3\in H^\nu(\R)$ for all $\nu\in[0,\tf{3}{2}-\tf{1}{2s})$. Then, $\widetilde{f}_3\in H^\mu(\R)$ for all $\mu\in[0,\tf{5}{2}-\tf{1}{2s})$, which clearly implies $\widetilde{f}_3\in H^{\f{3}{2}-\f{1}{4s}}(\R)$.
 
 Summing up, (observing that $f_3(0)=\green(0)$)
 \[
  f_2(t)=-\beta|q(0)|^{2\sigma}q(0)\green(0)-\beta|q(0)|^{2\sigma}q(0)b(s)t^{1-\f{1}{2s}}-\beta|q(0)|^{2\sigma}q(0)\widetilde{f}_3(t).
 \]
 In addition, we find that
 \[
  t^{1-\f{1}{2s}}=\f{2s-1}{2s}\int_0^t\f{1}{(t-\tau)^{\f{1}{2s}}}\dtau
 \]
 and that (recalling \eqref{eq-a})
 \[
  b(s)=\f{-\imath}{\pi}\int_0^\infty\int_0^1e^{-\imath z k^{2s}}\dz\dk=\f{-\imath}{\pi}\int_0^1z^{-\f{1}{2s}}\int_0^\infty e^{-\imath\rho^{2s}}\dro\dz=\f{-\imath a(s)2s}{2s-1}.
 \]
 Consequently, one can suitably rearrange the terms in \eqref{eq-charge_eq} in order to obtain
 \begin{equation}
  \label{eq-charge_new}
  q(t)=\widetilde{f}(t)\underbrace{-\imath a(s)\beta\int_0^t\f{|q(\tau)|^{2\sigma}q(\tau)-|q(0)|^{2\sigma}q(0)}{(t-\tau)^{\f{1}{2s}}}\dtau}_{=:\eta(t)},
 \end{equation}
 where now
 \[
  \widetilde{f}(t)=f_1(t)-\beta|q(0)|^{2\sigma}q(0)\green(0)-\beta|q(0)|^{2\sigma}q(0)\widetilde{f}_3(t)
 \]
 belongs to $H^{\f{3}{2}-\f{1}{4s}}(0,T)$.
 
 \emph{Part(ii): bootstrap argument}. Now, we can apply a bootstrap argument on \eqref{eq-charge_new} in order to obtain that $q\in H^{\f{3}{2}-\f{1}{4s}}(0,T)$.
 
 We know that the unique solution of $q$ of \eqref{eq-charge_new} belongs to $L^2(0,T)\cap C^0[0,T]$. Hence, from Lemma \ref{lem-ban_alg}, there results that $|q(\,\cdot\,)|^{2\sigma}q(\,\cdot\,)-|q(0)|^{2\sigma}q(0)\in L^2(0,T)\cap C^0[0,T]$ and (consequently), from Lemma \ref{lem-reg2} (item (i)), that $\eta \in H^{1-\f{1}{2s}}(0,T)\cap C^0[0,T]$, so that $q \in H^{1-\f{1}{2s}}(0,T)\cap C^0[0,T]$ (in turn).
 
 Repeating the same argument one can easily prove, with an iterative process, that $q \in H^{\f{3}{2}-\f{1}{4s}}(0,T)\cap C^0[0,T]$, which concludes the proof.
 
 However, some provisos are required:
 \begin{itemize}
  \item[(1)] one uses item (i) of Lemma \ref{lem-reg2} until the starting index of the iteration is smaller than $\f{1}{2}$ and item (ii) when the starting index becomes greater than $\f{1}{2}$;
  \item[(2)] if at some iteration one runs into $H^{\f{1}{2}}(0,T)$, which is not covered by Lemma \ref{lem-reg2}, then it is sufficient to observe that $q\in H^{\mu}(0,T)$, with $\mu=\f{1}{2}-\ep$ and (for instance) $\ep<\f{1}{2}-\f{1}{4s}$, so that one can use Lemma \ref{lem-reg2} to leap over $\f{1}{2}$ (since $\mu+1-\f{1}{2s}>\f{1}{2}$) and move on;
  \item[(3)] if at some iteration one find $q\in H^\mu(0,T)$ with $\mu\in(1,\f{3}{2}-\f{1}{4s})$, which is not covered by Lemma \ref{lem-ban_alg}, then it is sufficient to observe that $q\in H^{\f{1}{2}+\f{1}{4s}}(0,T)$, and since $\f{1}{2}+\f{1}{4s}<1$ one can use again Lemma \ref{lem-ban_alg} and Lemma \ref{lem-reg2} to obtain that $q\in H^{\f{3}{2}-\f{1}{4s}}(0,T)$.
 \end{itemize}
 Note also that the iterative process must end in a finite number of steps because the regularity gain at each step is always the same (as highlighted in Remark \ref{reg-gain}).
\end{proof}

We can now prove the first point of Theorem \ref{teo-nonlinear}. It is worth pointing out that the following proof holds for every $T\in(0,T^*)$.

\begin{proof}[Proof of Theorem \ref{teo-nonlinear}: item (i).]
 From Propositions \ref{pro-charge_cont} and \ref{prop-charge_reg}, there is $T\in(0,T^*)$ for which there exists a unique solution of \eqref{eq-charge_eq} belonging to $H^{\f{3}{2}-\f{1}{4s}}(0,T)$. As a consequence, the function
 \begin{equation}
  \label{eq-r}
  r(t):=\beta|q(t)|^{2\sigma}q(t)
 \end{equation}
 belongs to $H^1(0,T)$ (by Lemma \ref{lem-ban_alg}, observing that $\f{3}{2}-\f{1}{4s}>1$). Then, we can split the proof in three parts.
 
 \emph{Part 1): regularity of $\psi$.} Using the Fourier transform in \eqref{eq-ansatz} and recalling definition \eqref{eq-domn1}, we have
 
 \begin{equation}
 \label{eq-psi_tr}
  \tr{\psi}(t,k)=e^{-\imath|k|^{2s}t}\tr{\phi}_{\lambda,0}(k)-\f{r(0)e^{-\imath|k|^{2s}t}}{\sqrt{2\pi}(|k|^{2s}+\lambda)}-\f{\imath}{\sqrt{2\pi}}\int_0^te^{-\imath|k|^{2s}(t-\tau)}r(\tau)\dtau.
 \end{equation}
 In addition, an integration by parts yields
 \begin{equation}
  \label{eq-psibyparts}
  \tr{\psi}(t,k)=\underbrace{e^{-\imath|k|^{2s}t}\bigg\{\tr{\phi}_{\lambda,0}(k)+\f{1}{\sqrt{2\pi}(|k|^{2s}+\lambda)}\int_0^te^{\imath|k|^{2s}\tau}(\dot{r}(\tau)-\imath\lambda r(\tau))\dtau\bigg\}}_{=:\tr{\phi}_\lambda(t,k)}-\f{r(t)}{\sqrt{2\pi}(|k|^{2s}+\lambda)},
 \end{equation}
 so that
 \[
  \psi(t,x)=\phi_\lambda(t,x)-r(t)\green(x).
 \]
 Hence, in order to prove that $\psi(t,\,\cdot\,)\in\dom(\hamn)$ for all $t>0$ it suffices to prove that $\tr{\phi}_\lambda(t,\,\cdot\,)\in L^2(\R,|k|^{4s}\dk)$ (since it is clearly in $L^2(\R)$).
 
 First, we easily see that $e^{-\imath|k|^{2s}t}\tr{\phi}_{\lambda,0}\in L^2(\R,|k|^{4s}\dk)$, by the properties of the free propagator. Concerning the remaining part, as $r$ is trivially more regular than $\dot{r}$, it is sufficient to sow that
 \[
  \f{1}{|k|^{2s}+\lambda}\int_0^te^{\imath|k|^{2s}\tau}\dot{r}(\tau)\dtau\in L^2(\R,|k|^{4s}\dk),
 \]
 namely, that
 \[
  g(t,k):=\int_0^te^{\imath|k|^{2s}\tau}\dot{r}(\tau)\dtau\in L^2(\R)
 \]
 (as functions of $k$). Setting $\rho_t(\tau)=\dot{r}(\tau)\chi_{[0,t]}(\tau)$ and observing that $\tr{\rho}_t\in L^2(\R)\cap L^\infty(\R)$ (since $\rho_t\in L^2(\R)\cap L^1(\R)$), one obtains
 \begin{align*}
  \int_\R|g(t,k)|^2\dk= & \, 2\int_0^\infty\bigg|\int_0^t e^{\imath k^{2s}\tau}\dot{r}(\tau)\dtau\bigg|^2\dk=2\int_0^\infty\f{1}{\omega^{1-\f{1}{2s}}}\bigg|\int_0^t e^{\imath \omega\tau}\dot{r}(\tau)\dtau\bigg|^2\dome\\[.3cm]
                     = & \, 2\int_0^\infty\f{|\tr{\rho}_t(-\omega)|^2}{\omega^{1-\f{1}{2s}}}\dome \leq C(1+\sqrt{t})(1+\|\dot{r}\|_{L^2(0,t)}^2)
 \end{align*}
 (as $1-\f{1}{2s}\in(0,\f{1}{2}]$), which proves the claim. Furthermore, one can prove with analogous computations that
 \begin{equation}
  \label{eq-regphi}
  \phi_\lambda\in C^0([0,T];H^{2s}(\R))
 \end{equation}
 (for details see \cite{CFNT2}). However, if we endow the domain $\dom(\hamn)$ with the graph norm $\|\,\cdot\,\|_{\dom(\hamn)}$, then
 \[
  \|\psi\|_{\dom(\hamn)}:=\|\psi\|+\|\hamn \psi\|\leq C\left(\|\psi\|+\|(\hamz+\lambda)\phi_\lambda\|+|r|\right)\leq C\left(\|\phi_\lambda\|_{H^{2s}(\R)}+|q|^{2\sigma+1}\right)
 \]
 and thus, combining with \eqref{eq-regphi},
 \[
  \psi\in C^0([0,T];\dom(\hamn)).
 \]
 
 \emph{Part 2): regularity of $\f{\partial\psi}{\partial t}$.} Let us compute, then, the time derivative of $\tr{\psi}$. We have that
 \[
  \f{\partial\tr{\psi}}{\partial t}(t,k)=\f{\partial\tr{\phi}_{\lambda}}{\partial t}(t,k)-\dot{r}(t)\tr{\green}.
 \]
 In addition,
 \[
  \f{\partial\tr{\phi}_{\lambda}}{\partial t}(t,k)=-\imath|k|^{2s}\tr{\phi}_\lambda(t,k)+\f{\dot{r}(t)-\imath\lambda r(t)}{\sqrt{2\pi}(|k|^{2s}+\lambda)},
 \]
 so that
 \begin{equation}
  \label{eq-psider}
  \f{\partial\tr{\psi}}{\partial t}(t,k)=-\imath|k|^{2s}\tr{\phi}_\lambda(t,k)-\f{\imath\lambda r(t)}{\sqrt{2\pi}(|k|^{2s}+\lambda)}.
 \end{equation}
 As a consequence, arguing as in Part 1), one can easily see that
 \[
  \tr{\psi}\in C^1([0,T];L^{2}(\R)),
 \]
 which then proves that $\psi\in C^1([0,T];L^{2}(\R))$.
 
 \emph{Part 3): solution of \eqref{eq-cauchy}.}  Now, since the initial condition is clearly satisfied by construction, we have to prove that
 \begin{equation}
  \label{eq-NLSeq}
  \imath\f{\partial\psi}{\partial t}=\hamn\psi\qquad\text{in }L^2(\R), \qquad\forall t\in[0,T].
 \end{equation}
 By \eqref{eq-psider}, 
 \[
  \imath \tr{\f{\partial\psi}{\partial t}}(t,k)=|k|^{2s}\tr{\phi}_\lambda(t,k)+\lambda r(t)\tr{\green}(k).
 \]
 On the other hand,
 \[
  \tr{\hamn\psi}(t,k)=\tr{(\hamz+\lambda)\phi_\lambda}(t,k)-\lambda\tr{\psi}(t,k)=|k|^{2s}\tr{\phi}_\lambda(t,k)+\lambda(\tr{\phi}_\lambda(t,k)-\tr{\psi}(t,k)),
 \]
 which concludes the proof.
\end{proof}


\subsection{Conservation laws} This section is devoted to the proofs of the conservation laws associated to \eqref{eq-cauchy}. In particular, for the mass conservation we will prove that
\[
 \f{dM^2}{dt}(t)=0,\qquad\forall t\in[0,T^*),
\]
whereas we will show the energy conservation by a direct inspection of the equality
\[
 E(t)=E(0),\qquad\forall t\in[0,T^*).
\]

\begin{proof}[Proof of Theorem \ref{teo-nonlinear}: item (ii).]
 The proof can be divided in two parts.
 
 \emph{Part 1): mass conservation.} First we observe that
 \[
  \f{dM^2}{dt}(t)=2\mathrm{Re}\bigg\{\underbrace{\int_\R\ov{\psi(t,x)}\f{\partial\psi}{\partial t}(t,x)\dx}_{=:A}\bigg\},
 \]
 so that we have to prove that $A$ is purely imaginary.
 
 First, from \eqref{eq-NLSeq}
 \[
  \f{\partial\psi}{\partial t}(t,x)=-\imath\hamz\phi_\lambda(t,x)+\imath\lambda r(t)\green(x)
 \]
 (with $r$ defined by \eqref{eq-r}). As a consequence
 \[
  A=\imath\big(\underbrace{(\phi_\lambda(t,\,\cdot\,),\hamz\phi_\lambda(t,\,\cdot\,))-\lambda|r(t)|^2\|\green\|^2}_{=:A_1}\big)+\imath\big(\underbrace{\ov{r(t)}(\green,\hamz\phi_\lambda(t,\,\cdot\,))+r(t)(\phi_\lambda(t,\,\cdot\,),\green)}_{=:A_2}\big),
 \]
 where we easily see that, since $\hamz$ is self-adjoint, $A_1$ is real-valued. Now, a computation shows that
 \[
  \imath A_2=\imath\bigg(\underbrace{\ov{r(t)}(\green,(\hamz+\lambda)\phi_\lambda(t,\,\cdot\,)}_{=:B_1}-\underbrace{2\lambda\re{\ov{r(t)}(\green,\phi_\lambda(t,\,\cdot\,))}}_{=:B_2}\bigg).
 \]
 Finally, as $B_2$ is clearly real-valued, it is left to prove that $B_1$ is real-valued too. Exploiting the definition of the \rafr{Green's}{green} function \eqref{eq-green_eq} (which is also a real-valued function), the decomposition of the domain $\dom(\hamn)$ and \eqref{eq-r}, we have
 \[
  B_1=\ov{r(t)}((\hamz+\lambda)\green,\phi_\lambda(t,\,\cdot\,))=\ov{r(t)}\phi_\lambda(t,0)=|q(t)|^{2\sigma+2}+|r(t)|^2\green(0),
 \]
 that is in fact real-valued, which concludes the proof.
 
 \emph{Part 2): energy conservation.} From \eqref{eq-gagliardo_fou} and \eqref{eq-energy}, the energy at time $t$ reads
 \[
   E(t)=\int_\R|k|^{2s}|\tr{\psi}(t,k)|^{2}\dk+\frac{\beta}{\sigma+1}|q|^{2\sigma+2}.
 \]
 Using the representation of $\tr{\psi}(t,k)$ given by \eqref{eq-psi_tr}, the kinetic part of the energy turns out to be
 \begin{multline}
  \label{energia2}
  \int_\R|k|^{2s}|\widehat{\psi}(t,k)|^{2}\dk=\int_\R|k|^{2s}\left|\widehat{\psi}_{0}(t,k)-\frac{\imath}{\sqrt{2\pi}}\int_0^te^{\imath |k|^{2s}\tau}r(\tau)\dtau\right|^{2}\dk\\[.3cm]
  =\int_\R|k|^{2s}|\widehat{\psi}_{0}(t,k)|^{2}\dk+2\mathrm{Re}\bigg\{\underbrace{\frac{\imath}{\sqrt{2\pi}}\int_\R|k|^{2s}\widehat{\psi}_{0}(k)\int_{0}^{t}e^{-\imath |k|^{2s}\tau}\overline{r(\tau)}\dtau\dk}_{=:E_{1}}\bigg\}+\\
  +\underbrace{\frac{1}{2\pi}\int_\R|k|^{2s}\int_{0}^{t}\int_{0}^{t}e^{\imath |k|^{2s}(\tau_{1}-\tau_{2})}r(\tau_{1})\overline{r(\tau_{2})}\dtau_2\dtau_1\dk}_{=:E_{2}}.
 \end{multline}
 An integration by parts, yields
 \begin{align*}
  \label{eq-E1}
   E_1= & \, -\frac{1}{\sqrt{2\pi}}\int_\R\widehat{\psi}_{0}(k)\left(e^{-\imath|k|^{2s}t}\overline{r(t)}-\overline{r(0)}-\int_{0}^{t}e^{-\imath|k|^{2s}\tau}\overline{\dot{r}(\tau)}\dtau\right)\dk\\[.3cm]
     = &-\overline{r(t)}(\propz(t)\psi_0)(0)+\ov{r(0)}(\propz(0)\psi_0)(0)+\int_{0}^{t}\overline{\dot{r}(\tau)}(\propz(\tau)\psi_0)(0)\dtau,
 \end{align*}
 and recalling that $(\propz(0)\psi_0)(0)=\psi_0(0)=q(0)$,
 \begin{equation}
  \label{eq-E1}
  E_1=-\overline{r(t)}(\propz(t)\psi_0)(0)+\beta|q(0)|^{2\sigma+2}+\int_{0}^{t}\overline{\dot{r}(\tau)}(\propz(\tau)\psi_0)(0)\dtau.
 \end{equation}
 On the other hand, observing that
 \[
  \int_0^t\int_0^te^{\imath|k|^{2s}(\tau_1-\tau_2)}r(\tau_1)\ov{r(\tau_2)}\dtau_1\dtau_2=2\re{\int_0^t\int_0^{\tau_1}e^{\imath|k|^{2s}(\tau_1-\tau_2)}r(\tau_1)\ov{r(\tau_2)}\dtau_1\dtau_2},
 \]
 one has
 \[
  E_2=\f{1}{\pi}\re{\int_\R|k|^{2s}\int_0^te^{\imath|k|^{2s}\tau_1}r(\tau_1)\int_0^{\tau_1}e^{-\imath|k|^{2s}\tau_2}\ov{r(\tau_2)}\dtau_2\dtau_1\dk}.
 \]
 Now, another integration by parts shows that
 \begin{multline*}
  \int_0^t\imath|k|^{2s}e^{\imath|k|^{2s}\tau_1}r(\tau_1)\int_0^{\tau_1}e^{-\imath|k|^{2s}\tau_2}\ov{r(\tau_2)}\dtau_2\dtau_1\\[.2cm]
  =e^{\imath|k|^{2s}t}r(t)\int_0^te^{-\imath|k|^{2s}\tau}\ov{r(\tau)}\dtau-\int_0^te^{\imath|k|^{2s}\tau_1}\dot{r}(\tau_1)\int_0^{\tau_1}e^{-\imath|k|^{2s}\tau_2}\ov{r(\tau_2)}\dtau_2\dtau_1-\int_0^t|r(\tau)|^2\dtau,
 \end{multline*}
 so that (with some computations)
 \begin{equation}
  \label{eq-E2}
  E_2=2\re{\imath\ov{r(t)}\int_0^t\propz(t-\tau,0)r(\tau)\dtau+i\int_0^t\dot{r}(\tau_1)\int_0^{\tau_1}\propz(\tau_2-\tau_1,0)\ov{r(\tau_2)}\dtau_2\dtau_1}
 \end{equation}

Plugging \eqref{eq-E1} and \eqref{eq-E2} in (\ref{energia2}), we get
 \begin{multline*}
  \int_\R|k|^{2s}|\tr{\psi}(t,k)|^{2}\dk=\int_\R|k|^{2s}|\widehat{\psi}_{0}(k)|^{2}\dk+2\beta|q(0)|^{2\sigma+2}-2\re{\ov{r(t)}(\propz(t)\psi_0)(0)\dtau}\\[.3cm]
  +2\re{\int_0^t\ov{\dot{r}(\tau)}(\propz(\tau)\psi_0)(0)\dtau}+2\re{\imath\ov{r(t)}\int_0^t\propz(t-\tau,0)r(\tau)\dtau}+\\[.3cm]
  +2\re{\imath\int_0^t\dot{r}(\tau_1)\int_0^{\tau_1}\propz(\tau_2-\tau_1,0)\dtau_2\dtau_1}.
 \end{multline*}
 Consequently, from \eqref{eq-charge_def} and \eqref{eq-ansatz}, there results
 \[
  \ov{r(t)}\left(\imath\int_0^t\propz(t-\tau,0)r(\tau)-(\propz(t)\psi_0)(0)\right)=-\ov{r(t)}q(t)=-\beta|q(t)|^{2\sigma+2}
 \]
 and
 \[
  \int_0^t\ov{\dot{r}(\tau)}(\propz(\tau)\psi_0)(0)\dtau=\int_0^t\ov{\dot{r}(\tau)}q(\tau)\dtau+\ov{\bigg(-\imath\int_0^t\dot{r}(\tau)\int_0^\tau\propz(\eta-\tau,0)\ov{r(\eta)}\deta\dtau\bigg)},
 \]
 and thus (with a further integration by parts)
 \begin{multline*}
  \int_\R|k|^{2s}|\tr{\psi}(t,k)|^{2}\dk=\int_\R|k|^{2s}|\widehat{\psi}_{0}(k)|^{2}\dk+2\beta\big(|q(0)|^{2\sigma+2}-|q(t)|^{2\sigma+2}\big)+\\[.3cm]
  +2\re{\int_0^t\ov{\dot{r}(\tau)}q(\tau)\dtau}=-2\int_0^t\re{\ov{r(\tau)}\dot{q}(\tau)}\dtau.
 \end{multline*}
 Finally, as
 \[
  \f{\beta}{\sigma+1}\f{d}{d\tau}|q(\tau)|^{2\sigma+2}=2\beta|q(\tau)|^{2\sigma}\re{\ov{q(\tau)}\dot{q}(\tau)}=2\re{\ov{r(\tau)}\dot{q}(\tau)},
 \]
 then
 \[
  \int_\R|k|^{2s}|\tr{\psi}(t,k)|^{2}\dk=\int_\R|k|^{2s}|\widehat{\psi}_{0}(k)|^{2}\dk-\f{\beta}{\sigma+1}|q(t)|^{2\sigma+2}+\f{\beta}{\sigma+1}|q(0)|^{2\sigma+2},
 \]
 which (recalling \eqref{eq-charge_def}) concludes the proof.
\end{proof}


\subsection{Global well-posedness} Exploiting conservation laws, and in particular the energy conservation, it is possible to prove that the solution of \eqref{eq-cauchy} provided by \eqref{eq-ansatz}-\eqref{eq-charge_eq} is global in time, namely
\[
 T^*=+\infty
\]
(with $T^*$ defined by \eqref{eq-maximal}), in the defocusing and in the sub-critical/critical focusing cases.

In the focusing case the main ingredient is a fractional version of the Gagliardo-Nirenberg \rafr{inequality}{inequlity}, namely
\begin{equation}
 \label{eq-GN}
 \|f\|_{L^\infty(\R)}\leq C_s\|f\|^{1-\f{1}{2s}}[f]_{\dot{H}^s(\R)}^{\f{1}{2s}}
\end{equation}
(for the proof see, for instance, \cite{BM,E,MP}).

\begin{proof}[Proof of Theorem \ref{teo-nonlinear}: item (iii).]Consider, first, the defocusing case, i.e. $\beta>0$. From energy conservation, there results that
 \begin{equation}
  \label{eq-qlim}
  \limsup_{t\to T^*}|q(t)|=C<\infty
 \end{equation}
 and, by \cite[Theorem 2.3]{M}, this entails that $T^*=+\infty$.
 
 On the other hand, in the sub-critical focusing case, i.e. $\beta<0$ and $\sigma<\sigma_c(s)$, from \eqref{eq-GN} and \eqref{eq-charge_def},
 \begin{equation}
  \label{eq-coercivity}
  E(0)=E(t)\geq[\psi(t,\,\cdot\,)]_{\dot{H}^s(\R)}^2-\f{|\beta|}{\sigma+1}C_s^{2\sigma+2}\|\psi_0\|^{\f{(2s-1)(2\sigma+2)}{2s}}[\psi(t,\,\cdot\,)]_{\dot{H}^s(\R)}^{\f{2\sigma+2}{2s}}.
 \end{equation}
 Since $\tf{2\sigma+2}{2s}<2$ whenever $\sigma<\sigma_c(s)$, one obtains again \eqref{eq-qlim} and therefore the claim follows arguing as before.
 
 Finally, if $\sigma=\sigma_c(s)$, then \eqref{eq-coercivity} reads
 \begin{equation}
  \label{eq-epos}
  E(0)=E(t)\geq[\psi(t,\,\cdot\,)]_{\dot{H}^s(\R)}^2\left(1-\f{|\beta|C_s^{4s}}{2s}\|\psi_0\|^{2(2s-1)}\right)
 \end{equation}
 and hence \eqref{eq-qlim} is satisfied whenever the quantity in brackets is bigger than  $0$, namely whenever
 \[
  \|\psi_0\|<\left(\f{2s}{|\beta|C_s^{4s}}\right)^{\f{1}{2(2s-1)}}=:C(s,\beta),
 \]
 which concludes the proof.
\end{proof} 


\subsection{Blow-up solutions} In order to prove the rise of blow-up solutions, we use the classical Glassey method (see, e.g., \cite{G}) based on the definition of a \emph{moment of inertia} and on the proof of the so-called \emph{Virial Identity}.

Due to the different scaling properties of the fractional Laplacian, when $s<1$, it is necessary to slightly modify the standard definition of the moment of inertia. Precisely, we set
\begin{equation}
 \label{eq-Idef}
 I(t)=I(\psi(t,\,\cdot\,)):=\big\|(-\Delta)^{\f{1-s}{2}}x\psi(t,\,\cdot\,)\big\|^2=\big\||k|^{1-s}\partial_k\tr{\psi}(t,\,\cdot\,)\big\|^2
\end{equation}
where $\psi$ (henceforth) is the solution of \eqref{eq-cauchy} provided by \eqref{eq-charge_eq} and \eqref{eq-ansatz} (see for instance \cite{BHL} and the references therein).
\doma{We notice that the fractional momentum operator defined in the appendix in \cite{DdPDV15} bears some similarities with \eqref{eq-Idef}. 
It is unclear at the moment if there is a deeper connection between the two objects.}

The first point is to prove that $I$ is well defined on the maximal existence time of $\psi$, i.e. $[0,T^*)$ (with $T^*$ defined by \eqref{eq-maximal}).

\begin{lem}
 \label{lem-Iwell}
 Let $s\in(\f{1}{2},1]$, $\psi_0\in\dom(\hamn)$ and $I(\psi_0)<\infty$. Then, for every $T<T^*$,
 \begin{equation}
  \label{eq-Iwell}
  I(t)\leq C_T<\infty,\qquad\forall t\in[0,T].
 \end{equation}
\end{lem}

\begin{proof}
 Let $T<T^*$. \lorr{Recalling that $\tr{\psi}_0(k)=\tr{\phi}_{\lambda,0}-r(0)/\sqrt{2\pi}(|k|^{2s}+\lambda)$ and differentiating \eqref{eq-psi_tr} in $k$,}{From \eqref{eq-psi_tr}, an easy computation yields}
 \begin{multline*}
  |k|^{1-s}\f{\partial\tr{\psi}}{\partial k}(t,k)=-\underbrace{\imath2st|k|^s\sgn(k)\tr{\psi}(t,k)}_{=:A_1(t,k)}+\underbrace{\f{2s|k|^s\sgn(k)}{\sqrt{2\pi}}\int_0^te^{-\imath|k|^{2s}(t-\tau)}\tau r(\tau)\dtau}_{=:A_2(t,k)}+\\[.3cm]
  +\underbrace{e^{-\imath|k|^{2s}t}|k|^{1-s}\f{\partial\tr{\psi}_0}{\partial k}(k)}_{=:A_3(t,k)}.
 \end{multline*}
 Now, as $\psi\in C^0([0,T];\dom(\hamn))$, $\psi\in C^0([0,T];H^s(\R))$ as well, so that
 \[
  \|A_1(t,\,\cdot\,)\|^2\leq C T\|\psi\|_{C^0([0,T];H^s(\R))}^2<\infty.
 \]
 On the other hand
 \[
  \|A_3(t,\,\cdot\,)\|^2= I(0)<\infty,
 \]
 by assumption. It is, then, left to discuss $A_2$. An integration by parts shows that
 \[
  A_2(t,k)=\f{\imath2s|k|^s\sgn(k)}{\sqrt{2\pi}(|k|^{2s}+\lambda)}\left(-r_1(t)+\int_0^te^{-\imath|k|^{2s}(t-\tau)}(\dot{r}_1(\tau)-\imath\lambda r_1(\tau))\dtau\right)
 \]
 with $r_1(\tau):=\tau r(\tau)$ (which is clearly at least as regular as $r(\tau)$). Consequently, arguing as in the proof of Theorem \ref{teo-nonlinear} (precisely, item (i), Part 1)), one immediately sees that 
 \eqref{eq-Iwell} is satisfied.
\end{proof}

As a second point, we have to prove the fractional  Virial Identity.

\begin{pro}
 \label{pro-first_der}
 Let $s\in(\f{1}{2},1]$, $\psi_0\in\dom(\hamn)$ and $I(\psi_0)<\infty$. Then, $I(\,\cdot\,)\in C^1[0,T^*)$ and
 \begin{equation}
  \label{eq-first_der}
  \dot{I}(t)=4s\im{\int_\R k\tr{\psi}(t,k)\ov{\f{\partial\tr{\psi}}{\partial k}(t,k)}\dk},\qquad\forall t\in[0,T^*).
 \end{equation}
\end{pro}

\begin{proof}
 We start by computing the derivative of the integrand of $I(t)$ (given by \eqref{eq-Idef}), namely $A(t,k):=|k|^{2-2s}|\partial_k\tr{\psi}(t,k)|^2$. First we note that
 \begin{equation}
  \label{eq-A}
  \f{\partial A}{\partial t}(t,k)=2|k|^{2-2s}\re{\ov{\f{\partial\tr{\psi}}{\partial k}(t,k)}\f{\partial^2\tr{\psi}}{\partial t\partial k}(t,k)}.
 \end{equation}
 Now, recalling that
 \begin{multline}
  \label{eq-dekpsi}
  \f{\partial\tr{\psi}}{\partial k}(t,k)=-\imath2st|k|^{2s-1}\sgn(k)\tr{\psi}(t,k)+\\[.3cm]
  +\f{2s|k|^{2s-1}\sgn(k)}{\sqrt{2\pi}}\int_0^te^{-\imath|k|^{2s}(t-\tau)}\tau r(\tau)\dtau+e^{-\imath|k|^{2s}t}\f{\partial\tr{\psi}_0}{\partial k}(k)
 \end{multline}
 and, \lora{differentiating in $t$, with some computations there results}
 \begin{multline*}
  \lora{\f{\partial}{\partial t}\f{\partial\tr{\psi}}{\partial k}(t,k)=-\imath|k|^{2s-1}\sgn(k)\tr{\psi}(t,k)}+\\[.3cm]
  \lora{-\imath|k|^{2s}\bigg(e^{-\imath|k|^{2s}t}\f{\partial\tr{\psi}_0}{\partial k}(k)+\f{2s|k|^{2s-1}\sgn{k}}{\sqrt{2\pi}}\int_0^te^{-\imath|k|^{2s}(t-\tau)}\tau r(\tau)\dtau\bigg)}\\[.3cm]
  \lora{\f{2s|k|^{2s-1}\sgn(k)tr(t)}{\sqrt{2\pi}}-\imath2st|k|^{2s-1}\sgn(k)\f{\partial\tr{\psi}}{\partial t}(t,k).}
 \end{multline*}
 \lora{Then, since by \eqref{eq-psider}}
 \[
  \f{\partial\tr{\psi}}{\partial t}(t,k)=-\imath|k|^{2s}\tr{\psi}(t,k)-\f{\imath r(t)}{\sqrt{2\pi}},
 \]
 we find
 \begin{multline*}
  \lora{\f{\partial}{\partial t}\f{\partial\tr{\psi}}{\partial k}(t,k)=-\imath|k|^{2s-1}\sgn(k)\tr{\psi}(t,k)}+\\[.3cm]
  \lora{-\imath|k|^{2s}\bigg(-\imath2st|k|^{2s-1}\sgn(k)\tr{\psi}(t,k)+\f{2s|k|^{2s-1}\sgn{k}}{\sqrt{2\pi}}\int_0^te^{-\imath|k|^{2s}(t-\tau)}\tau r(\tau)\dtau+e^{-\imath|k|^{2s}\tau}\f{\partial\tr{\psi}_0}{\partial k}(k)\bigg)}
 \end{multline*}
 \lora{and, hence, using again \eqref{eq-dekpsi},}
 \[
  \f{\partial^2\tr{\psi}}{\partial t\partial k}(t,k)=-\imath 2s|k|^{2s-1}\sgn(k)\tr{\psi}(t,k)-\imath|k|^{2s}\f{\partial\tr{\psi}}{\partial k}(t,k).
 \]
 Thus, plugging into \eqref{eq-A}, there results
 \begin{equation}
  \label{eq-A2}
  \f{\partial A}{\partial t}(t,k)=4s\im{k\tr{\psi}(t,k)\ov{\f{\partial\tr{\psi}}{\partial k}(t,k)}}.
 \end{equation}
 
 On the other hand, fix an arbitrary $T<T^*$. An integration by parts in \eqref{eq-dekpsi} shows that
 \begin{multline*}
  k\f{\partial\tr{\psi}}{\partial k}(t,k)=-\imath2st|k|^{2s}\sgn(k)\tr{\phi}_\lambda(t,k)+\\[.3cm]
  +\f{\imath2s|k|^{2s}\sgn(k)}{\sqrt{2\pi}(|k^{2s}+\lambda|)}\int_0^te^{-\imath|k|^{2s}(t-\tau)}(\dot{r}_1(\tau)-i\lambda r_1(\tau))\dtau+e^{-\imath|k|^{2s}t}k\f{\partial\tr{\psi}_0}{\partial k}(k),
 \end{multline*}
 (where again $r_1(\,\cdot\,):=\tau r(\,\cdot\,)\in H^1(0,T)$). Hence
 \[
  \bigg|k\f{\partial\tr{\psi}}{\partial k}(t,k)\bigg|\leq C_T\bigg(|k|^{2s}|\phi_\lambda(t,k)|+1+\bigg|k\f{\partial\tr{\psi}_0}{\partial k}(k)\bigg|\bigg),\qquad\forall t\in[0,T].
 \]
 In addition, since by \eqref{eq-psibyparts}
 \[
  |\tr{\psi}(t,k)|\leq C_T\left(|\phi_\lambda(t,k)|+\f{1}{|k|^{2s}+\lambda}\right),\qquad\forall t\in[0,T],
 \]
 there results
 \begin{multline*}
  \left|\f{\partial A}{\partial t}(t,k)\right|\leq C_T\bigg(|k|^{2s}|\tr{\phi}_\lambda(t,k)|^2+|\tr{\phi}_\lambda(t,k)|+\f{1}{|k|^{2s}+\lambda}+\bigg.\\[.3cm]
  +|\tr{\phi}_\lambda(t,k)|\bigg|k\f{\partial\tr{\psi}_0}{\partial k}(k)\bigg|+\f{|k|}{|k|^{2s}+\lambda}\bigg|\f{\partial\tr{\psi}_0}{\partial k}(k)\bigg|\bigg),\qquad\forall t\in[0,T].
 \end{multline*}
Now as,
 \[
  |\tr{\phi}_\lambda(t,k)|\leq C_T\bigg(|\tr{\phi}_{\lambda,0}(k)|+\f{1}{|k|^{2s}+\lambda}\bigg),\qquad\forall t\in[0,T],
 \]
there results
 \begin{multline*}
  \left|\f{\partial A}{\partial t}(t,k)\right|\leq C_T\bigg(|k|^{2s}|\tr{\phi}_{\lambda,0}(k)|^2+|\tr{\phi}_{\lambda,0}(k)|+\f{1}{|k|^{2s}+\lambda}+\bigg.\\[.3cm]
  +|\tr{\phi}_{\lambda,0}(k)|\bigg|k\f{\partial\tr{\psi}_0}{\partial k}(k)\bigg|+\f{|k|}{|k|^{2s}+\lambda}\bigg|\f{\partial\tr{\psi}_0}{\partial k}(k)\bigg|\bigg),\qquad\forall t\in[0,T],
 \end{multline*}
 which is clearly integrable by the regularity of $\phi_{\lambda,0}$ and since $I(\psi_0)<\infty$. Hence, by \eqref{eq-A2} and dominated convergence, one gets \eqref{eq-first_der} (and the continuity of $\dot{I}$ on $[0,T^*)$).
\end{proof}

As a third point, we can compute the second derivative of the moment of inertia.

\begin{pro}
 \label{pro-second_der}
 Let $s\in(\f{1}{2},1]$ and $\psi_0\in\dom(\hamn)$ with $\phi_{\lambda,0}\in\mathcal{S}(\R)$. Then, $I(\,\cdot\,)\in C^2[0,T^*)$ and
 \begin{equation}
  \label{eq-second_der}
  \ddot{I}(t)=8s^2E(0)+\f{4s\beta\big(\sigma-\sigma_c(s)\big)}{\sigma+1}|q(t)|^{2\sigma+2},\qquad\forall t\in[0,T^*).
 \end{equation}
\end{pro}

\begin{proof}
 Fix an arbitrary $T<T^*$ and focus on the interval $[0,T]$. Preliminarily, we observe that the assumptions entail that $I(\psi_0)<\infty$ (so that Lemma \ref{lem-Iwell} and Proposition \ref{pro-first_der} are valid). On the other hand, setting
 \[
  B(t,k):=\im{k\tr{\psi}(t,k)\ov{\f{\partial\tr{\psi}}{\partial k}(t,k)}},
 \]
 one sees that
 \begin{equation}
  \label{eq-derB}
  \f{\partial B}{\partial t}(t,k)=2s|k|^{2s}|\tr{\psi}(t,k)|^2-\f{1}{\sqrt{2\pi}}\re{r(t)k\ov{\f{\partial\tr{\psi}}{\partial k}(t,k)}}.
 \end{equation}
 Arguing as in the proof of Proposition \ref{pro-first_der}, for every fixed $R>0$,
 \[
  \bigg|\f{\partial B}{\partial t}(t,k)\bigg|\leq g_{R,T}(k)\in L^1(-R,R),\qquad\forall t\in[0,T],
 \]
 so that, if we set
 \[
  \dot{I}_R(t):=4s\int_{-R}^RB(t,k)\dk,
 \]
 then by dominated convergence
 \[
  \ddot{I}_R(t)=4s\int_{-R}^R\f{\partial B}{\partial t}(t,k)\dk,
 \]
 which is, in addition, a continuous function in $[0,T]$. As a consequence, \lorr{if one can prove that $\ddot{I}_R$ converges pointwise a.e. in $[0,T]$ and that $|\ddot{I}_R(t)|\leq f_T(t)$, for a.e. $t\in[0,T]$, with $f_T(t)\in L^1(0,T)$,}{whether one can prove that $\ddot{I}_R$ converges pointwise a.e. in $[0,T]$ and that $|\ddot{I}_R(t)|\leq f_T(t)\in L^1(0,T)$, for a.e. $t\in[0,T]$,} then (since clearly $\dot{I}_R(t)\to\dot{I}(t)$ pointwise) by dominated convergence there results
 \begin{equation}
  \label{eq-dom_conv}
  \dot{I}(t)=\dot{I}(0)+\lim_{R\to\infty}\int_0^t\ddot{I}_R(\tau)\dtau=\dot{I}(0)+\int_0^t\lim_{R\to\infty}\ddot{I}_R(\tau)\dtau,
 \end{equation}
 whence
 \begin{equation}
  \label{eq-der_lim}
  \lim_{R\to\infty}\ddot{I}_R(\tau)=\ddot{I}(\tau).
 \end{equation}
 
 From \eqref{eq-derB}
 \[
  \ddot{I}_R(\tau)=\underbrace{8s^2\int_{-R}^R|k|^{2s}|\tr{\psi}(\tau,k)|^2\dk}_{=:\Phi(\tau,R)}-\underbrace{\f{4s}{\sqrt{2\pi}}\re{r(\tau)\int_{-R}^Rk\ov{\f{\partial\tr{\psi}}{\partial k}(\tau,k)}\dk}}_{=:\Psi(\tau,R)}.
 \]
 First we see that $\Phi(\tau,R)\to8s^2[\psi(\tau,\,\cdot\,)]_{H^s(\R)}^2$ and that $|\Phi(\tau,R)|\leq C\|\psi\|_{C^0([0,T];H^s(\R))}$, for every $\tau\in[0,T]$. Furthermore, an integration by parts yields
 \[
  \int_{-R}^Rk\f{\partial\tr{\psi}}{\partial k}(\tau,k)\dk=\underbrace{R(\tr{\psi}(\tau,R)+\tr{\psi}(\tau,-R))}_{=:A_1(\tau,R)}-\underbrace{\int_{-R}^R\tr{\psi}(t,k)\dk}_{=:A_2(\tau,R)}.
 \]
 Now,
 \[
  A_1(\tau,R)\to0,\qquad\text{as}\quad R\to\infty,\quad\forall \tau\in[0,T],
 \]
 and
 \[
  |A_1(\tau,R)|\leq C_T,\qquad\forall R>0,\quad\forall \tau\in[0,T],
 \]
 from \eqref{eq-psibyparts} and the assumptions on $\phi_{\lambda,0}$. On the other hand,
 \[
  A_2(\tau,R)\to\sqrt{2\pi}q(\tau),\qquad\text{as}\quad R\to\infty,\quad\forall \tau\in[0,T]
 \]
 and
 \[
  |A_2(\tau,R)|\leq C\|\psi\|_{C^0([0,T];H^s(\R))},\qquad\forall R>0,\quad\forall \tau\in[0,T].
 \]
 Thus we can pass to the limit in \eqref{eq-dom_conv} by dominated convergence, and from \eqref{eq-der_lim} we obtain
 \[
  \ddot{I}(\tau)=8s^2[\psi(t,\,\cdot\,)]_{H^s(\R)}^2+4s\beta|q(\tau)|^{2\sigma+2}.
 \]
 Finally, this immediately implies that $\ddot{I}$ is continuous on $[0,T]$ and, with some easy computations, that \eqref{eq-second_der} is satisfied.
\end{proof}

\begin{rem}
 \label{rem-second_der}
 In the proof of Proposition \ref{pro-second_der}, the only point where the assumption $\phi_{\lambda,0}\in\mathcal{S}(\R)$ is required is in the discussion of $A_1(\tau,R)$. It is then clear that it is not the minimal one. We refer to Section \ref{subsec:main} for the reason of such a choice.
\end{rem}

Finally, we can show the proof of item (iv) of Theorem \ref{teo-nonlinear}.

\begin{proof}[Proof of Theorem \ref{teo-nonlinear}: item (iv).]
 Let $\beta<0$. Hence \eqref{eq-second_der} reads
 \[
  \ddot{I}(t)=8s^2E(0)-\f{4s|\beta|\big(\sigma-\sigma_c(s)\big)}{\sigma+1}|q(t)|^{2\sigma+2},\qquad\forall t\in[0,T^*).
 \]
 If $\sigma\geq\sigma_c(s)$ and $E(0)<0$, then $I$ is uniformly concave in $[0,T^*)$, namely
 \[
  \ddot{I}(t)\leq C<0,\qquad\forall t\in[0,T^*).
 \]
 As a consequence
 \[
  I(t)\leq I(0)+\dot{I}(0)t+\f{Ct^2}{2},\qquad\forall t\in[0,T^*).
 \]
 Assume, therefore, by contradiction that $T^*=+\infty$. Then $\lim_{t\to T^*}\lora{I(t)}=-\infty$ but this is prevented by the fact that $I(t)\geq0$, for all $t\in[0,T^*)$.
\end{proof}


\subsection{Stationary states} This last part of the paper is devoted to the proof of Theorem \ref{teo-standing}. Preliminarily, we recall some computations that descend from some easy changes of variables and \cite[Eq. 3.194.3]{GR}.

Let $\omega>0$ and $s\in(\f{1}{2},1]$. Denoting by $B(\,\cdot\,,\,\cdot\,)$ the Beta function \rafa{(\cite{Er})}
\[
 \rafa{B(x,y)=\int_{0}^{1}t^{x-1}(1-t)^{y-1}dt, \quad \textrm{Re}\,x>0,\,\,\textrm{Re}\,y>0, }
 \]
and by $\Gamma(\,\cdot\,)$ the Euler Gamma function \rafa{(\cite{Er})}
\[
 \rafa{\Gamma(z)=\int_{0}^{\infty}e^{-t}t^{z-1}dt, \quad \textrm{Re}\,z>0,}
 \]
 there results
\begin{multline}
 \label{eq-conto1}
 \f{1}{\pi}\int_0^\infty\f{1}{k^{2s}+\omega}\dk=\f{1}{2\pi s\omega}\int_0^\infty\f{\eta^{\f{1}{2s}-1}}{1+\f{\eta}{\omega}}\deta=\f{\omega^{\f{1}{2s}}B(\f{1}{2s},1-\f{1}{2s})}{2\pi s\omega}=\\[.3cm]
 =\f{\omega^{\f{1}{2s}-1}\Gamma(\f{1}{2s})\Gamma(1-\f{1}{2s})}{2\pi s}=\f{\omega^{\f{1}{2s}-1}}{2s\sin(\f{\pi}{2s})}>0.
\end{multline}
and
\begin{multline}
 \label{eq-conto2}
 \f{1}{\pi}\int_0^\infty\f{1}{(k^{2s}+\omega)^2}\dk=\f{1}{2\pi s\omega^2}\int_0^\infty\f{\eta^{\f{1}{2s}-1}}{(1+\f{\eta}{\omega})^2}\deta=\f{\omega^{\f{1}{2s}}B(\f{1}{2s},2-\f{1}{2s})}{2\pi s\omega^2}=\\[.3cm]
 =\f{\omega^{\f{1}{2s}-2}\Gamma(\f{1}{2s})\Gamma(2-\f{1}{2s})}{2\pi s}=\f{(2s-1)\omega^{\f{1}{2s}-2}\Gamma(\f{1}{2s})\Gamma(1-\f{1}{2s})}{4\pi s^2}=\f{(2s-1)\omega^{\f{1}{2s}-2}}{4s^2\sin(\f{\pi}{2s})}>0.
\end{multline}

\begin{proof}[Proof of Theorem \ref{teo-standing}]
 We divide the proof in three parts.
 
 \emph{Part 1): $\omega>0$.} Assume that $u^\omega$ is a standing wave. As it belongs to $\dom(\hamn)$,
 \[
  u^\omega=\phi_\lambda^\omega-\underbrace{\beta u^\omega(0)|u^\omega(0)|^{2\sigma}}_{=:r^\omega}\green,\qquad\forall \lambda>0,
 \]
 and, since it must satisfy \eqref{eq-stationary}, there results
 \begin{equation}
  \label{eq-standing_2}
  (\flap+\lambda)\phi_\lambda^\omega=(\lambda-\omega)u^\omega,
 \end{equation}
 whence
 \[
  (\flap+\omega)\phi_\lambda^\omega=-(\lambda-\omega)r^\omega\green.
 \]
 Now, by means of the Fourier transform, the previous equality reads
 \[
  \tr{\phi}_\lambda^\omega=\f{r^\omega(\omega-\lambda)}{\sqrt{2\pi}(|k|^{2s}+\lambda)(|k|^{2s}+\omega)}
 \]
 and, since $r^\omega\neq0$ and $\phi_\lambda^\omega\in H^{2s}(\R)$ for all $\lambda>0$, there results $\omega>0$.
 
 \emph{Part 2): proof of \eqref{eq-standing}.} As $\omega>0$, let us choose $\omega=\lambda$, so that \eqref{eq-standing_2} reads
 \[
  (\flap+\omega)\phi_\omega^\omega=0\qquad\Leftrightarrow\qquad\phi_\omega^\omega\equiv0.
 \]
 As a consequence, $u^\omega$ has to satisfy
 \begin{equation}
  \label{eq-standing_3}
  u^\omega(x)=-\beta u^\omega(0)|u^\omega(0)|^{2\sigma}\greeno(x)
 \end{equation}
 and thus
 \[
  u^\omega(0)=-\beta u^\omega(0)|u^\omega(0)|^{2\sigma}\greeno(0),
 \]
 or, equivalently (since $u^\omega(0)\neq0$),
 \begin{equation}
  \label{eq-standing_zero}
  1=-\beta|u^\omega(0)|^{2\sigma}\greeno(0).
 \end{equation}
 Since $\greeno(0)>0$, clearly if $\beta>0$, then \eqref{eq-standing_zero} cannot be satisfied. Therefore, there cannot exist any standing wave in the defocusing case.
 
 Let us consider the focusing case (where \eqref{eq-standing_zero} can be fulfilled). It is clear that, up to the multiplication of a constant phase factor, $u^\omega(0)>0$ and that
 \begin{equation}
  \label{eq-u_zero}
  u^\omega(0)=\left(\f{1}{|\beta|\greeno(0)}\right)^{\f{1}{2\sigma}}.
 \end{equation}
 Now, recalling that
 \begin{equation}
  \label{eq-green_zero}
  \greeno(0)=\f{1}{2\pi}\int_\R\f{1}{|k|^{2s}+\omega}\dk=\f{1}{\pi}\int_0^\infty\f{1}{k^{2s}+\omega}\dk
 \end{equation}
 and combining with \eqref{eq-conto1} and \eqref{eq-standing_3}, \eqref{eq-standing} follows.
 
 \emph{Part 3): proof of \emph{(i)}, \emph{(ii)} and \emph{(iii)}.} The energy of a standing wave (in the focusing case) is given by
 \[
  E(u^\omega)=\underbrace{[u^\omega]_{\dot{H}^{s}(\R)}^2}_{=:K(u^\omega)}-\underbrace{\frac{|\beta|}{\sigma+1}|u^\omega(0)|^{2\sigma+2}}_{=:P(u^\omega)}.
 \]
 Combining \eqref{eq-u_zero}, \eqref{eq-green_zero} and \eqref{eq-conto1}, with some computations one sees that
 \[
  P(u^\omega)=\f{(2s\sin(\f{\pi}{2s}))^{1+\f{1}{\sigma}}\omega^{\f{(\sigma+1)(2s-1)}{2s\sigma}}}{|\beta|^{\f{1}{\sigma}}(\sigma+1)}.
 \]
 On the other hand, combining \eqref{eq-standing_3}, \eqref{eq-conto2}, \eqref{eq-green_zero} and \eqref{eq-u_zero}
 \begin{align*}
  C(u^\omega)= & \, \f{|\beta|^2|u^\omega(0)|^{4\sigma+2}}{2\pi}\int_R\f{|k|^{2s}}{(|k|^{2s}+\omega)^2}\dk\\[.3cm]
             = & \, |\beta|^2|u^\omega(0)|^{4\sigma+2}\left(\greeno(0)-\f{\omega}{\pi}\int_R\f{1}{(|k|^{2s}+\omega)^2}\dk\right)\\[.3cm]
             = & \, |\beta|^2|u^\omega(0)|^{4\sigma+2}\left(\greeno(0)-\f{(2s-1)\omega^{\f{1}{2s}-1}}{4s^2\sin(\f{\pi}{2s})}\right)\\[.3cm]
             = & \, |\beta|^2|u^\omega(0)|^{4\sigma+2}\f{\omega^{\f{1}{2s}-1}}{4s^2\sin(\f{\pi}{2s})}=\f{(2s)^{\f{1}{\sigma}}(\sin(\f{\pi}{2s}))^{1+\f{1}{\sigma}}\omega^{\f{(\sigma+1)(2s-1)}{2s\sigma}}}{|\beta|^{\f{1}{\sigma}}}.
 \end{align*}
 Summing up,
 \[
  E(u^\omega)=\f{(2s)^{\f{1}{\sigma}}(\sin(\f{\pi}{2s}))^{1+\f{1}{\sigma}}\omega^{\f{(\sigma+1)(2s-1)}{2s\sigma}}}{|\beta|^{\f{1}{\sigma}}}\left(1-\f{2s}{\sigma+1}\right),
 \]
 which clearly proves (i), (ii) and (iii).
\end{proof}


\appendix
\section{Proof of Proposition \ref{pro-linear}}
\label{sec:linear}

In order to prove Proposition \ref{pro-linear} some further information on the regularity properties of Green's function is required.

\begin{lem}
 Let $s\in(\f{1}{2},1]$ and $\lambda>0$. Then
 \begin{gather}
  \label{eq-green_prop1}
  D^{2s-1}\green\in H^1(\R\backslash\{0\}),\\[.3cm]
  \label{eq-green_prop2}
  [D^{2s-1}\green](0)=-1,\\[.3cm]
  \label{eq-green_prop3}
  \green(0)=\lambda\|\green\|^2+\|\flaph\green\|^2.
 \end{gather}
\end{lem}

\begin{proof}
 We divide the proof in three parts.
 
 \emph{Part (i): proof of \eqref{eq-green_prop1}}. Combining \eqref{eq-green} and \eqref{eq-der} yields
 \[
  D^{2s-1}\green(x)=\f{\imath}{2\pi}\int_\R e^{\imath kx}\f{|k|^{2s-1}\sgn(k)}{|k|^{2s}+\lambda}\dk,
 \]
 which clearly belongs to $L^2(\R)$. Hence, one can easily check (using the Fourier transform) that
 \[
  \f{d}{dx}D^{2s-1}\green=\f{1}{\sqrt{2\pi}}(\lambda\green-\delta)\qquad\text{in}\quad\mathcal{D}'(\R),
 \]
 so that
 \[
  (D^{2s-1}\green,\varphi')=-\lambda(\green,\varphi)\qquad\forall\varphi\in C_0^\infty(\R\backslash\{0\}),
 \]
 which then proves \eqref{eq-green_prop1}.
 
 \emph{Part (ii): proof of \eqref{eq-green_prop2}}. Let us compute, then,
 \[
  D^{2s-1}\green(0^+):=\lim_{x\downarrow0}\f{\imath}{2\pi}\int_\R e^{\imath kx}\f{|k|^{2s-1}\sgn(k)}{|k|^{2s}+\lambda}\dk.
 \]
 First, for every $x>0$, setting $p:=xk$ and observing that
 \[
  \int_{|p|\leq1}\f{|p|^{2s}\sgn(p)}{|p|\left(|p|^{2s}+\lambda x^{2s}\right)}\dpi=0,
 \]
 one obtains
 \begin{multline*}
  D^{2s-1}\green(x)=\f{\imath}{2\pi}\int_{|p|\leq1}\f{e^{ip}|p|^{2s}}{p\left(|p|^{2s}+\lambda x^{2s}\right)}\dpi+\f{\imath}{2\pi}\int_{|p|>1}\f{e^{ip}|p|^{2s}}{p\left(|p|^{2s}+\lambda x^{2s}\right)}\dpi\\[.3cm]
  = \underbrace{\f{\imath}{2\pi}\int_{|p|\leq1}\f{(e^{ip}-1)|p|^{2s}}{p\left(|p|^{2s}+\lambda x^{2s}\right)}\dpi}_{=:I_1}+\underbrace{\f{\imath}{2\pi}\int_{|p|>1}\f{e^{ip}}{p}\dpi}_{I_2}-\underbrace{\f{\imath\lambda x^{2s}}{2\pi}\int_{|p|>1}\f{e^{ip}}{p\left(|p|^{2s}+\lambda x^{2s}\right)}}_{I_3}.
 \end{multline*}
 Now, clearly $I_2$ is independent of $x$ and $s$ and is finite as a Fresnel integral (see, for details, \cite[Eqs. 5.2.1 and 5.2.2]{AS} and \cite[Eqs. 3.722.1 and 3.722.3]{GR}). Furthermore,
 \[
  \left|\f{(e^{ip}-1)|p|^{2s}}{p\left(|p|^{2s}+\lambda x^{2s}\right)}\right|\leq C,\qquad\forall p\in[-1,1],
 \]
 and
 \[
  \left|\f{e^{ip}}{p\left(|p|^{2s}+\lambda x^{2s}\right)}\right|\leq\f{C}{|p|^{2s+1}}\in L^1(\R\backslash[-1,1]),\qquad\forall p\in\R\backslash[-1,1],
 \]
 so that
 \[
  I_1\to\f{\imath}{2\pi}\int_{|p|\leq1}\f{e^{ip}-1}{p}\dpi,\qquad I_3\to0,\qquad\text{as}\quad x\downarrow0,
 \]
 which are independent of $s$ as well. Consequently, 
 \[
  D^{2s-1}\green(0^+)=\lim_{x\downarrow0}\left.D^{2s-1}\green(x)\right|_{s=1}=\lim_{x\downarrow0}\frac{d}{dx}\mathcal{G}_1^\lambda(x)=-\frac{1}{2}.
 \]
 In the very same way one can prove that $D^{2s-1}\green(0^-)=\frac{1}{2}$, and thus \eqref{eq-green_prop2} follows immediately.
 
 \emph{Part (iii): proof of \eqref{eq-green_prop3}}. First we note that, from \eqref{eq-flaptr} and \eqref{eq-green}, $\flaph\green\in L^2(\R)$. In addition, by definition one finds
 \[
  \green(0)=(\green,(\hamz+\lambda)\green)=\lambda\|\green\|^2+(\green,\hamz\green)=\lora{\lambda\|\green\|^2+}(\flaph\green,\flaph\green),
 \]
 which then concludes the proof.
\end{proof}

\begin{proof}[Proof of Proposition \ref{pro-linear}]
 The proof can be divided in two parts.
 
 \emph{Part (i): proof of \eqref{eq-dom2} and \eqref{eq-act2}.} First we focus on the inclusion
 \begin{equation}
  \label{eq-inc1}
  \dom(\ham)\subset\left\{\psi\in H^s(\R):D^{2s-1}\psi\in H^1(\R\backslash\{0\}),\,[D^{2s-1}\psi](0)=\alpha\psi(0)\right\}.
 \end{equation}
 If $\psi\in\dom(\ham)$, then
 \[
  \psi(x)=\phi_\lambda(x)-\alpha\psi(0)\green(x),\qquad\phi_\lambda\in H^{2s}(\R),\quad\lambda>0.
 \]
 As a consequence, since $\green\in H^s(\R)$, one immediately finds that $\psi\in H^s(\R)$. On the other hand, as $D^{2s-1}\phi_\lambda\in H^1(\R)$, recalling \eqref{eq-green_prop1} and \eqref{eq-green_prop2}, one obtains that $D^{2s-1}\psi\in H^1(\R\backslash\{0\})$ and that
 \begin{equation}
  \label{eq-salti}
  [D^{2s-1}\psi](0)=-\alpha\psi(0)[D^{2s-1}\green](0)=\alpha\psi(0),
 \end{equation}
 thus proving \eqref{eq-inc1}.
 
 On the other hand, in order to prove
 \begin{equation}
  \label{eq-inc2}
  \dom(\ham)\supset\left\{\psi\in H^s(\R):D^{2s-1}\psi\in H^1(\R\backslash\{0\}),\,[D^{2s-1}\psi](0)=\alpha\psi(0)\right\}
 \end{equation}
 it is sufficient to show that, if $\psi$ belongs to the r.h.s. of \eqref{eq-inc2}, then
 \[
  \phi_\lambda:=\psi+\alpha\psi(0)\green\in H^{2s}(\R).
 \]
 Preliminarily, we note that $\phi_\lambda\in H^s(\R)$ and that $D^{2s-1}\phi_\lambda\in H^1(\R\backslash\{0\})$. However, \eqref{eq-salti} immediately entails that $[D^{2s-1}\phi_\lambda](0)=0$ and hence $D^{2s-1}\phi_\lambda\in H^1(\R)$, which completes the proof.
 
 Finally, one easily sees that for $x\neq0$, $\hamz\green=-\lambda\green$, and thus
 \[
  \ham\psi=(\hamz+\lambda)\phi_\lambda-\lambda\psi=\hamz\phi_\lambda+\alpha\lambda\psi(0)\green=\hamz\psi,
 \]
 which proves \eqref{eq-act2}.
 
 \emph{Part(ii): proof of \eqref{eq-form}}. From \eqref{eq-green_eq} and \eqref{eq-dom1}, with some (easy) computations one has
 \[
  (\psi,\ham\psi)=(\phi_\lambda,\hamz\phi_\lambda)+2\alpha\lambda\re{\ov{\psi(0)}(\green,\phi_\lambda)}-\alpha\phi_\lambda\ov{\psi(0)}-\alpha^2\lambda|\psi(0)|^2\|\green\|^2
 \]
 for all $\psi\in\dom(\ham)$. On the other hand, we first observe that $\|\flaph\psi\|^2<\infty$ as $\psi\in H^s(\R)$, and then, arguing in an analogous way, there results
 \begin{multline*}
  \|\flaph\psi\|^2+\alpha|\psi(0)|^2=(\phi_\lambda,\hamz\phi_\lambda)+2\alpha\lambda\re{\ov{\psi(0)}(\green,\phi_\lambda)}\\[.3cm]
  -2\alpha\re{\phi_\lambda(0)\ov{\psi(0)}}+\alpha^2|\psi(0)|^2\|\flaph\green\|^2+\alpha|\psi(0)|^2.
 \end{multline*}
 As a consequence, recalling that the boundary conditions imply
 \[
  \phi_\lambda(0)=(1+\alpha\green(0))\psi(0),
 \]
 one sees that \eqref{eq-form} is satisfied if and only if
 \[
  \green(0)|\psi(0)|^2-\lambda\|\green\|^2|\psi(0)|^2=\|\flaph\green\|^2|\psi(0)|^2\qquad\forall\psi\in\dom(\ham).
 \]
 However, since this is clearly true by means of \eqref{eq-green_prop3}, one obtains that \eqref{eq-form} holds for all $\psi\in\dom(\ham)$. Finally, one can easily check that the set of the functions in $L^2(\R)$ such that $\form(\psi)<\infty$ is $H^s(\R)$, thus concluding the proof.
\end{proof}



\begin{thebibliography}{99}

\bibitem{AS}
M. Abramovitz, I.A. Stegun, 
\emph{Handbook of Mathematical Functions: with Formulas, Graphs, and Mathematical Tables},
National Bureau of Standards Applied Mathematics Series 55, U.S. Government Printing Office, Washington, D.C. 1964

\bibitem{ACCT}
R. Adami, R. Carlone, M. Correggi, L. Tentarelli,
Blow-up for the pointwise NLS in dimension two: absence of critical power,
\emph{preprint}, arXiv:1808.10343 [math.AP] (2018).

\bibitem{ADFT}
R. Adami, G. Dell'Antonio, R. Figari, A. Teta,
The Cauchy problem for the Schr\"odinger equation in dimension three with concentrated nonlinearity,
\emph{Ann. Inst. H. Poincar\'e Anal. Non Lin\'eaire} \textbf{20} (2003), no. 3, 477--500.

\bibitem{ADFT2}
R. Adami, G. Dell'Antonio, R. Figari, A. Teta,
Blow-up solutions for the Schr\"odinger equation in dimension three with a concentrated nonlinearity,
\emph{Ann. Inst. H. Poincar\'e Anal. Non Lin\'eaire} \textbf{21} (2004), no. 1, 121--137.

\bibitem{AT}
R. Adami, A. Teta,
A class of nonlinear Schr\"odinger equations with concentrated nonlinearity,
\emph{J. Funct. Anal.} \textbf{180} (2001), no. 1, 148--175. 

\bibitem{AGH-KH}
S. Albeverio, F. Gesztesy, R. H{\o}egh-Krohn, H. Holden,
\emph{Solvable models in quantum mechanics},
Texts and Monographs in Physics, Springer-Verlag, New York, 1988.

\bibitem{BE}
J. Bertoin,
\emph{L\'evy processes},
Cambridge Tracts in Mathematics 121, Cambridge University Press, Cambridge, 1996. 

\bibitem{BHL}
T. Boulenger, D. Himmelsbach, E. Lenzmann,
Blowup for fractional NLS,
\emph{J. Funct. Anal.} \textbf{271} (2016), no. 9, 2569--2603. 

\bibitem{BM}
H. Brezis, P. Mironescu,
Gagliardo-Nirenberg, composition and products in fractional Sobolev spaces,
\emph{J. Evol. Equ.} \textbf{1} (2001), no. 4, 387--404.

\bibitem{CCNP}
C. Cacciapuoti, R. Carlone, D. Noja, A. Posilicano, 
The one-dimensional Dirac equation with concentrated nonlinearity,
\emph{SIAM J. Math. Anal.} \textbf{49} (2017), no. 3, 2246--2268.

\bibitem{CFNT}
C. Cacciapuoti, D. Finco, D. Noja, A. Teta,
The NLS equation in dimension one with spatially concentrated nonlinearities: the pointlike limit,
\emph{Lett. Math. Phys.} \textbf{104} (2014), no. 12, 1557--1570.

\bibitem{CFNT2}
C. Cacciapuoti, D. Finco, D. Noja, A. Teta,
The point-like limit for a NLS equation with concentrated nonlinearity in dimension three,
\emph{J. Funct. Anal.} \textbf{273} (2017), no. 5, 1762--1809.

\bibitem{CCF}
R. Carlone, M. Correggi, R. Figari,
Two-dimensional time-dependent point interactions,
\emph{Functional analysis and operator theory for quantum physics}, 189--211, EMS Ser. Congr. Rep., Eur. Math. Soc., Z\"urich, 2017.

\bibitem{CCT}
R. Carlone, M. Correggi, L. Tentarelli, Well-posedness of the two-dimensional nonlinear Schr\"odinger equation with concentrated nonlinearity, \emph{Ann. Inst. H. Poincar\'e Anal. Non Lin\'eaire} \textbf{36} (2019), no. 1, 257--294.

\bibitem{CFN}
R. Carlone, R. Figari, C. Negulescu,
The quantum beating and its numerical simulation,
\emph{J. Math. Anal. Appl.} \textbf{450} (2017), no. 2, 1294--1316.

\bibitem{CFT2}
R. Carlone, A. Fiorenza, L. Tentarelli,
The action of Volterra integral operators with highly singular kernels on H\"older continuous, Lebesgue and Sobolev functions,
\emph{J. Funct. Anal.} \textbf{273} (2017), no. 3, 1258--1294.

\bibitem{CHH1}
Y. Cho, H. Hajaiej, G. Hwang, T. Ozawa,
On the Cauchy problem of fractional Schr\"odinger equation with Hartree type nonlinearity,
\emph{Funkcial. Ekvac.} \textbf{56} (2013), no. 2, 193--224.  

\bibitem{CHKL1}
Y. Cho, G. Hwang, S. Kwon, S. Lee,
Profile decompositions and blowup phenomena of mass critical fractional Schr\"odinger equations,
\emph{Nonlinear Anal.} \textbf{86} (2013), 12--29.  

\bibitem{CHKL2}
Y. Cho, G. Hwang, S. Kwon, S. Lee,
On finite time blow-up for the mass-critical Hartree equations,
\emph{Proc. Roy. Soc. Edinburgh Sect. A} \textbf{145} (2015), no. 3, 467--479.

\bibitem{DdPDV15}
J. D\'avila, M. del Pino, S. Dipierro, E. Valdinoci, 
Concentration phenomena for the nonlocal Schr\"odinger equation with Dirichlet datum, 
\emph{Anal. PDE} \textbf{8} (2015), no. 5, 1165--1235. 

\bibitem{DSV17}
S. Dipierro, O. Savin, E. Valdinoci, 
All functions are locally s-harmonic up to a small error, 
\emph{J. Eur. Math. Soc.} \textbf{19} (2017), no. 4, 957--966.

\bibitem{DV}
S. Dipierro, E. Valdinoci,
A Simple Mathematical Model Inspired by the Purkinje Cells: From Delayed Travelling Waves to Fractional Diffusion,
\emph{Bull. Math. Biol.} online (2018), DOI: 10.1007/s11538-018-0437-z.

\bibitem{DPV}
E. Di Nezza, G. Palatucci, E. Valdinoci,
Hitchhiker's guide to the fractional Sobolev spaces,
\emph{Bull. Sci. Math.} \textbf{136} (2012), no. 5, 521--573.

\bibitem{Er}
A. Erd\'{e}lyi, W. Magnus, F. Oberhettinger, F. G.  Tricomi,      
\emph {Higher transcendental functions. {V}ol. {III}}, Robert E. Krieger Publishing Co., Inc., Melbourne, Fla., 1981

\bibitem{E}
A. Esfahani,
Anisotropic Gagliardo-Nirenberg inequality with fractional derivatives,
\emph{Z. Angew. Math. Phys.} \textbf{66} (2015), no. 6, 3345--3356.

\bibitem{G}
R.T. Glassey,
On the blowing up of solutions to the Cauchy problem for nonlinear Schr\"odinger equations,
\emph{J. Math. Phys.} \textbf{18} (1977), no. 9, 1794--1797.

\bibitem{GR}
I.S. Gradshteyn, I.M. Ryzhik,
\emph{Tables of integrals, series and products},
Elsevier/Academic Press, Amsterdam, 2007.

\bibitem{GH}
B. Guo, Z. Huo,
Global well-posedness for the fractional nonlinear Schr\"odinger equation,
\emph{Comm. Partial Differential Equations} \textbf{36} (2011), no. 2, 247--255.

\bibitem{HY}
Y. Hong, Y. Sire,
On fractional Schr\"odinger equations in Sobolev spaces,
\emph{Commun. Pure Appl. Anal.} \textbf{14} (2015), no. 6, 2265--2282.

\bibitem{IP}
A.D. Ionescu, F. Pusateri,
Nonlinear fractional Schr\"odinger equations in one dimension,
\emph{J. Funct. Anal.} \textbf{266} (2014), no. 1, 139--176.

\bibitem{KZ}
K. Kirkpatrick and Y. Zhang,
Fractional Schr\"odinger dynamics and decoherence,
\emph{Physica D} \textbf{332} (2016), 41--54.

\bibitem{KLR}
J. Krieger, E. Lenzmann, Pierre Rapha\"el,
Nondispersive solutions to the $L^{2}$-critical half-wave equation,
\emph{Arch. Ration. Mech. Anal.} \textbf{209} (2013), no. 1, 61--129.

\bibitem{KP}
A. Kufner, L.E. Persson,
\emph{Weighted inequalities of Hardy type},
World Scientific Publishing Co., Inc., River Edge, NJ, 2003.

\bibitem{LA}
N.S. Landkof,
\emph{Foundations of modern potential theory},
Die Grundlehren der mathematischen Wissenschaften Band 180, Springer-Verlag, New York-Heidelberg, 1972.      

\bibitem{L1}
N. Laskin,
Fractional quantum mechanics and L\'evy path integrals,
\emph{Phys. Lett. A} \textbf{268} (2000), no. 4-6, 298--305.

\bibitem{L2}
N. Laskin,
Fractional Schr\"odinger equation,
\emph{Phys. Rev. E (3)} \textbf{66} (2002), no. 5, 056108, 7 pp.

\bibitem{LRSRM}
E.K. Lenzi, H.V. Ribeiro, M.A.F. dos Santos, R. Rossato, R.S. Mendes,
Time dependent solutions for a fractional Schr\"odinger equation with delta potentials,
\emph{J. Math. Phys.} \textbf{54} (2013), no. 8, 082107, 8 pp.

\bibitem{L}
S. Longhi,
Fractional Schr\"odinger equation in optics,
Opt. Lett. \textbf{40} (2015), no. 6, 1117--1120.

\bibitem{MV}
A. Massaccesi, E.Valdinoci,
Is a nonlocal diffusion strategy convenient for biological populations in competition?,
\emph{J. Math. Biol.} \textbf{74} (2017), no. 1--2, 113--147.

\bibitem{MOS}
A. Michelangeli, A. Ottolini, R. Scandone,
Fractional powers and singular perturbations of quantum differential Hamiltonians, \emph{J. Math. Phys.} \textbf{59} (2018), no. 7, 072106, 27 pp.

\bibitem{MS}
A. Michelangeli, R. Scandone,
Point-like perturbed fractional laplacians through shrinking potentials of finite range,
\emph{preprint}, arXiv:1803.10191 [math.FA] (2018).

\bibitem{M}
R. K. Miller,
\emph{Nonlinear Volterra integral equations},
Mathematics Lecture Note Series, W. A. Benjamin, Inc., Menlo Park, Calif., 1971.

\bibitem{MP}
C. Morosi, L. Pizzocchero,
On the constants for some fractional Gagliardo-Nirenberg and Sobolev inequalities,
\emph{Expo. Math.} \textbf{36} (2018), no. 1, 32--77.

\bibitem{PV}
S. Patrizi, E. Valdinoci,
Relaxation times for atom dislocations in crystals,
\emph{Calc. Var. Partial Differential Equations} \textbf{55} (2016), no. 3, Art. 71, 44 pp.

\bibitem{S}
A. Sacchetti,
Stationary solutions of a fractional Laplacian with singular perturbation,
\emph{preprint} arXiv:1801.01694, [math-ph] (2018).

\bibitem{UB}
N. Uzar, S. Ballikaya,
Investigation of classical and fractional Bose-Einstein condensation for harmonic potential,
\emph{Physica A} \textbf{392}, (2013), no. 8, 1733--1741.

\end{thebibliography}
\end{document}